\documentclass[10pt,journal]{IEEEtran}
\usepackage{amsmath}
\usepackage{amssymb}
\usepackage{mathrsfs}
\usepackage{amsthm}
\usepackage{multirow}
\usepackage{color}
\usepackage{array}
\usepackage{url}
\usepackage{comment}
\usepackage{enumerate}
\usepackage{eucal}
\usepackage[normalem]{ulem}
\usepackage{graphicx}
\usepackage{tabto}
\usepackage{amsmath}
\usepackage{caption}
\usepackage{subcaption}
\usepackage{bbm}
\usepackage{amsthm}
\usepackage{amssymb}
\usepackage{enumitem}
\usepackage{enumerate}
\usepackage{mathtools}
\usepackage{esvect}
\usepackage{caption}
\usepackage{float}
\usepackage{multicol}
\usepackage{tikz}
\usepackage{soul}

\usepackage[style=ieee]{biblatex}
\usepackage{balance}
\usepackage{adjustbox}

\usepackage{hyperref}
\newcommand\myshade{70}
\hypersetup{
	linkcolor  = red!\myshade!black,
	citecolor  = blue!\myshade!black,
	urlcolor   = blue!\myshade!black,
	colorlinks = true,
}

\usepackage{cleveref}

\usepackage{tabularx}
\usepackage{algorithm}
\usepackage[noend]{algpseudocode}

\usepackage{tikz}
\usetikzlibrary{arrows, patterns, shapes.arrows, decorations.pathmorphing}

\IEEEoverridecommandlockouts

\theoremstyle{plain}
\newtheorem{theorem}{Theorem}

\newtheorem{lemma}[theorem]{Lemma}
\newtheorem{proposition}[theorem]{Proposition}

\theoremstyle{definition}
\newtheorem{definition}{Definition}
\newtheorem{example}[definition]{Example}

\newtheorem{remark}[definition]{Remark}

\newcommand{\FF}{\mathbb{F}}
\newcommand{\BB}{\mathbb{B}}

\DeclareMathAlphabet{\mathbfsl}{OT1}{ppl}{b}{it} 

\newcommand{\vy}{\mathbfsl{y}}

\newcommand{\vc}{\mathbfsl{c}}

\newcommand{\cS}{\mathcal{S}}
\newcommand{\cA}{\mathcal{A}}

\newcommand{\cC}{\mathcal{C}}

\newcommand{\cF}{\mathcal{F}}

\newcommand{\cT}{\mathcal{T}}

\newcommand{\cR}{\mathcal{R}}
\newcommand{\cRT}{\mathcal{RT}}
\usepackage{balance}

\newcommand{\vlambda}{{\pmb{\lambda}}}

\newcommand{\rs}{{\rm RS}}
\newcommand{\grs}{{\rm GRS}}
\newcommand{\GF}{{\rm GF}}
\newcommand{\tr}{{\rm Tr}}
\newcommand{\supp}{{\rm supp}}
\newcommand{\wt}{{\rm wt}}

\newcommand{\floor}[1]{{\left\lfloor #1\right\rfloor}}
\newcommand{\ceil}[1]{{\left\lceil #1\right\rceil}}

\newcolumntype{Y}{>{\centering\arraybackslash}X} 

\newcommand{\hanmao}[1]{{\color{magenta} (HM: #1)}}

\newcommand{\kim}[1]{{\color{blue} (Kim: #1)}}

\title{Robust Repair of Reed-Solomon Codes}

\author{
    Wilton Kim\IEEEauthorrefmark{1}, 
    Stanislav Kruglik\IEEEauthorrefmark{2}, 
    Gaojun Luo\IEEEauthorrefmark{3}, 
    and Han Mao Kiah\IEEEauthorrefmark{1}%
    
    \thanks{\IEEEauthorrefmark{1}Wilton Kim and Han Mao Kiah are with the School of Physical and Mathematical Sciences, Nanyang Technological University, Singapore (email: \{wilton.kim, hmkiah\}@ntu.edu.sg).}
    \thanks{\IEEEauthorrefmark{2}Stanislav Kruglik is with the Department of Electrical and Photonics Engineering, Technical University of Denmark, Kongens Lyngby, Denmark (email: stakr@dtu.dk).}
    \thanks{\IEEEauthorrefmark{3}Gaojun Luo is with the School of Mathematics, Nanjing University of Aeronautics and Astronautics, Nanjing, China (email: gaojun\_luo@nuaa.edu.cn). Gaojun Luo’s contribution to this work was carried out while he was a Research Fellow at Nanyang Technological University, Singapore.}
    \thanks{The work of Wilton Kim and Han Mao Kiah was supported by the Ministry of Education, Singapore, under its MOE AcRF Tier 1 Award under Grant 24/25}
    \thanks{This paper was presented in part at the 2023 IEEE International Symposium on Information Theory~\cite{kruglik2023} DOI: 10.1109/ISIT54713.2023.10206718 and 2024 IEEE International Symposium on Information Theory~\cite{kim2024decoding} DOI: 10.1109/ISIT57864.2024.10619533.}
     \thanks{\textit{Corresonding author:} Stanislav Kruglik}
}

\begin{document}
\date{}

\maketitle
\begin{abstract}
    We study the problem of robust repair of a single erasure in Reed--Solomon codes under low communication bandwidth. 
Focusing on the Guruswami--Wootters trace repair framework, we investigate whether a failed node can be correctly repaired in the presence of erroneous responses from helper nodes. 
Equivalently, we view the collection of downloaded traces as a code, which we call the repair-trace code. By characterizing the zero coefficients of the associated polynomial in terms of cyclotomic cosets, we derive upper bounds on the dimension $k$ that allow correction of a given number of erroneous traces $e$, as well as lower bounds on the minimum distance as a function of $k$. For the case $q=2$, we exploit explicit formulas for cyclotomic coset representatives to obtain the exact optimal dimension bound for single-error correction. We also propose two efficient robust repair schemes. 
Our first scheme achieves the error-correction capability guaranteed by the BCH bound. 
To approach a stronger bound based on character sums, we develop a second scheme that tolerates more errors at the cost of an additional factor $n$ in computational complexity.
\end{abstract}


\section{Introduction}\label{sec:intro}

Modern distributed storage systems store massive amounts of data across many servers, and server failures are routine rather than exceptional. 
To protect data against such failures, these systems commonly use erasure-correcting codes, with Reed--Solomon codes~\cite{originalRSpaper} being among the most widely deployed in practice (see~\cite{Dinh2022} for a survey). 
Reed--Solomon codes enjoy the maximum distance separable (MDS) property: given a codeword stored across $n$ nodes, any $k$ nodes suffice to recover the full data, so the system can tolerate up to $n-k$ simultaneous node failures. 
When a storage node fails, the system must reconstruct the missing content by contacting the remaining nodes. 
In practice, as documented by Rashmi \textit{et al.}~\cite{rashmi2013solution}, single-node failure is by far the most common case, making its efficient repair a central concern. 
Because communication is often a major bottleneck in large-scale storage systems, a central objective is to minimize the amount of data transferred during this recovery process. 
This communication cost, known as the {\em repair bandwidth}, was formalized by Dimakis \textit{et al.}~\cite{Dimakis2010}; a detailed and comprehensive surveys can be found in~\cite{overview2, 10.1145/3708994}.

For Reed--Solomon codes defined over $\FF = \GF(p^{mt})$, Guruswami and Wootters later showed that the repair bandwidth for single-node failure can be significantly reduced by downloading carefully chosen traces (sub-symbols over a smaller base field $\BB = \GF(p^m)$) instead of full symbols in $\FF$~\cite{guruswamiwooters2017}. 
The key insight is to contact more than $k$ helper nodes while downloading only a single trace from each, which can yield significant bandwidth savings over naive repair. 
This approach, known as the {\em trace repair framework}, was subsequently extended to many settings including multiple erasures and repair of linear combinations of erased nodes~\cite{Duursma2017, Dau2017, Kiah2018, Kiah2021, ShuttyWootters2022, Tamo2022i, Tamo2022ii, kiah2023, 10614076, 10533730}. 
However, all of these works assume that helper nodes return correct responses, which may fail in practice due to hardware faults, stale data, or adversarial behavior.

A natural question is whether the bandwidth savings of the trace repair framework can be retained when some helper nodes provide incorrect traces. 
Previous work~\cite{chenbarg2020} on erroneous-node repair for Reed--Solomon codes primarily address constructions in high sub-packetization regimes (e.g., schemes attaining the cut-set bound), which typically require very large field sizes and therefore high storage per node. 
In contrast, our focus is on the small- to moderate-field regime that is more relevant for practical implementations of Reed--Solomon-based storage. 
In this regime, the error-correcting capability of the trace repair mechanism itself has not been studied.

To address this, we focus on the Guruswami--Wootters repair scheme (which we review in Section~\ref{sec::tracerepairframework}) for a full-length Reed-Solomon code and ask the following question. 

\textit{Can we correctly repair a failed node in the presence of erroneous traces through the trace repair framework?}

In this paper, we consider a full-length Reed-Solomon code, the collection of traces downloaded during repair as a {\em $\BB$-linear code}. For a fixed erased node $\alpha^*\in\FF$, we call such a code the \textit{repair-trace} code $\cRT(\FF\setminus \{\alpha^*\},k)$ which we formally define in Definition~\ref{def:repairtracecode}. 
Specifically, we answer the following questions:
\begin{enumerate}[label = (\Alph*)]
    \item Fix the number of erroneous traces $e$. What is the largest dimension $k$ of the Reed--Solomon code for which the trace repair framework guarantees correct recovery in the presence of up to $e$ erroneous traces? 
    \item Fix the dimension $k$ of the Reed-Solomon code. What is a lower bound of the minimum distance of the corresponding repair-trace code? Or equivalently, how many erroneous traces can we tolerate in the repair-trace process?
\end{enumerate}
Furthermore, to practically realize these robust repair scheme, we also design efficient decoders that correct $e$ erroneous traces and, as a result, ensure correct repair. 

\subsection{Our Contributions}

We begin by showing that the problem of repairing any single erased node can be reduced to the problem of repairing $f(0)$. Therefore, it suffices to study $\cRT(\FF^*,k)$, where $\FF^* = \FF\setminus\{0\}$. To facilitate our analysis of the repair-trace code, we study the support $\supp(F)$, the set of the positions of nonzero coefficients of the corresponding polynomial $F(x)$ in $\cRT(\FF^*,k)$. We show that $\supp(F)$ is contained in a union of some cyclotomic cosets, which we refer to as $\cS_k$ (see Theorem~\ref{thm:sup_is_union_of_cyc_coset}). Then, we embed $\cRT(\FF^*,k)$ into an $\FF$-linear supercode $\cT(\FF^*,k)$ which consists of all polynomials whose support is contained within $\cS_k$, formally defined in Theorem~\ref{thm:sup_is_union_of_cyc_coset}. In fact, the repair-trace code $\cRT(\FF^*,k)$ is a subfield subcode of $\cT(\FF^*,k)$ over the base field $\BB$, i.e., $\cRT(\FF^*,k) = \cT(\FF^*,k) \cap \BB^n$. We show that $\cT(\FF^*,k)$ is a cyclic code (see Theorem~\ref{thm:rep_trace_is_cyc}) and utilize this property for our further analysis.

\vspace{2mm}

\noindent\underline{\textit{Dimension Upper Bounds for Error Correction.}} The cyclic property of $\cT(\FF^*,k)$ allows us to apply the BCH bound which guarantees error correction of $e$ given that there are $2e$ $b$-spaced consecutive zero coefficients with $\gcd(b,q^t-1) = 1$. These zero coefficients correspond to the cyclotomic cosets excluded from $\cS_k$. As established in Theorem~\ref{thm:sup_is_union_of_cyc_coset}, the set $\cS_k$ always contains the cyclotomic coset with the largest representative and it is constructed in a bottom-up manner, i.e., we add more cyclotomic cosets to $\cS_k$ as $k$ grows. Hence, to obtain the dimension upper-bound, we introduce a greedy pruning algorithm (Algorithm~K) that systematically removes the second largest representative coset (while keeping the largest) in order to determine the largest dimension $k$, say when $k = K$, that produces a gap of $2e$ $b$-spaced consecutive zero coefficients. Equivalently, the code $\cT(\FF^*,K)$ has minimum distance of at least $2e+1$. Since $\cRT(\FF^*,K)\subset \cT(\FF^*,K)$, the code $\cRT(\FF^*,K)$ also has a minimum distance of at least $2e+1$ and guarantees error correction of $e$ errors. 
Thus, the dimension upper bound $K$ also applies to the repair-trace code.

Although the dimension upper bound can be obtained by applying Algorithm~K, expressing it explicitly for a general $q,t$ and $e$ is not trivial. Nevertheless, for $\FF = \GF(2^t)$, we derive explicit formulas for the coset representatives and obtain the largest admissible dimension that guarantees single erroneous trace correction. We summarize it in the following theorem.

\begin{theorem}\label{thm:binary_bound}
Let $\mathbb{F} = \mathrm{GF}(2^t)$ with $t\ge 3$. The repair-trace code $\cRT(\FF^*,k)$ has minimum distance at least $3$ (guaranteeing single-error correction) for any $k \le K$, where
\begin{equation*}
K =
\begin{cases}
2^{t-1} - 2^{\lfloor (t-1)/2 \rfloor}, & \text{if } t \text{ is odd}, \\[4pt]
2^{t-1} - 2^{\lfloor (t-1)/2 \rfloor + 1}, & \text{if } t \text{ is even}.
\end{cases}
\end{equation*}
Furthermore, this upper bound $K$ is optimal: if $k = K+1$, the minimum distance of $\cRT(\cA,k)$ is at most $2$.
\end{theorem}

\noindent\underline{\textit{Lower Bounds for the Minimum Distance.}}
To study the number of erroneous traces the trace-repair framework tolerates, we establish three lower bounds for the minimum distance of the repair-trace code $\cRT(\FF^*,k)$. The first lower bound is inherited naturally due to the cyclic property of $\cT(\FF^*,k)$, which yields the BCH bound. By fixing the starting point of the run of consecutive zero terms, we obtain a closed-form lower bound, which we call the Degree bound. Since $\cRT(\FF^*,k)\subset \cT(\FF^*,k)$, the distance lower bounds established for $\cT(\FF^*,k)$ also hold for $\cRT(\FF^*,k)$. 

Crucially, the BCH bound and Degree bound are established for a denser supercode $\cT(\FF^*,k)$. Hence, they often severely underestimate the resilience of the actual repair-trace code $\cRT(\FF^*,k)$. In response to this, we derive a third lower bound, the character sum bound, which is established directly for $\cRT(\FF^*,k)$. We summarize all three lower bounds for the minimum distance in the following theorem.

\begin{theorem}[Distance Lower Bounds]\label{thm:distance_lower_bound}
Let $\mathbb{F} = \mathrm{GF}(p^{mt})$ with $\BB=\GF(p^m)$. 
Suppose that the repair-trace code $\cRT(\FF^*,k)$ has minimum distance $d$.
Then the following lower bounds for $d$ applies.
\begin{enumerate}[label = (\roman*),leftmargin=1.5em]
    \item \textit{(BCH Bound)}. Let $\mathcal{S}_k$ be the support as defined in Theorem~\ref{thm:sup_is_union_of_cyc_coset}. If there exist integers $a$ and $b$ with $\gcd(b, n) = 1$ such that $(a + i b) \pmod n \notin \mathcal{S}_k$ for all $0 \le i \le \delta-1$, then the minimum distance satisfies $d \ge d_{\text{BCH}} = \delta + 1$.
    \item \textit{(Degree Bound)}. If $k\leq p^m$, then $d\ge p^{mt}-1-\Delta$, where
\begin{equation*}\label{eq:Delta}
    \Delta \triangleq 
    \begin{cases}
        (k-1)p^{mt-m}, &\text{ when }k\ge 2,\\
        p^{mt-m}-1, &\text{ when }k= 1.
    \end{cases}
\end{equation*}
   \item \textit{(Character Sum Bound)}. If $k<1+\frac{p^{mt}-1} {\sqrt{p^{mt}}}$,then $d\ge d_{\text{CS}}$, where
    \begin{equation*}\label{eq:d3}
        \hspace{0cm}d_{\text{CS}} \triangleq 
        \begin{cases}
            \frac{p^m-p}{p^m}\left(p^{mt}-1-(k-1)\sqrt{p^{mt}}\right),\text{when }m\ge 2,\\
            \frac{p-1}{p}\left(p^{t}-1-(k-1)\sqrt{p^{t}}\right),\hspace{0.7cm}\text{when }m = 1.
        \end{cases}
    \end{equation*}
\end{enumerate}
\end{theorem}

\noindent\underline{\textit{Robust Repair Schemes.}} 
Finally, using insights from our theoretical bounds, we design efficient robust repair schemes that recover the lost symbol directly from the downloaded traces, in the presence of erroneous responses.
First, we introduce our baseline approach, \textbf{Robust Repair Scheme 1}, that recovers $f(0)$ directly while achieving the error-correcting capability guaranteed by the BCH bound. We begin by applying a specific bijective transformation to the received word that mirrors the algebraic shifts used in the BCH bound derivation. 
This transformation allows us to apply the Berlekamp--Welch algorithm \cite{guruswami2012} to decode up to a radius of $\lfloor\delta/2\rfloor$. Applying the inverse transformation, the output is mapped back to a polynomial corresponding to the supercode $\cT(\FF^*, k)$. Crucially, we prove that if $e \le \lfloor\delta/2\rfloor$, this scheme outputs a valid polynomial in $\cRT(\FF^*, k)$ and correctly repairs $f(0)$. 

As shown in Fig.~\ref{fig:comparison_of_decoders_and_distance_bound}, there remains a gap in between the number of erroneous traces correctable by Robust Repair Scheme 1 and the number guaranteed by the Character Sum bound. 
To approach the Character Sum bound, we introduce \textbf{Robust Repair Scheme 2}. The core idea is to systematically reduce the support by evaluating candidate values for certain coefficients of $f(x)$. As shown later in Theorem~\ref{thm:sup_is_union_of_cyc_coset}, the support of the repair-trace polynomial $F(x)=\tr(f(x)/x)$ is formed by a union of cyclotomic cosets. We observe that some specific cosets within this union correspond to exactly one coefficient of $f(x)$. 
Suppose that $C_*$ is such coset corresponding to coefficient $f_*$ and has coset representative $r_*$.
For a guess $f^*$, we subtract the corresponding terms from the received word $\boldsymbol{y} = (y_{\alpha})_{\alpha\in\FF^*}$, i.e., $y_{\alpha,f_*} = y_\alpha - \tr(f_* \alpha^{r_{*}})$ for each $\alpha\in\FF^*$. 
For the correct guess $f_*$, this transformation is a translation and we preserve the number of errors. Furthermore, for the correct guess $f_*$, the polynomial interpolating the uncorrupted symbols in $\boldsymbol{y}_{f_*}$ has a smaller support. This enlarges the pool of zero coefficients which can improve the BCH bound. Since the true $f_*$ is unknown, we iteratively test all candidates.
Specifically, for each candidate $f_*$, we apply the same transformation and decoding procedure in Robust Repair Scheme 1, albeit with a larger decoding radius. 
Incorrect guesses of $f_*$ are rejected by a consistency ensuring the the output polynomial lies in $\cRT(\FF^*,k)$. As we can see in Fig.~\ref{fig:comparison_of_decoders_and_distance_bound}, Robust Repair Scheme 2 achieves stronger error correction at the cost of an additional factor $n$ in computational complexity.

However, as we can see in Fig.~\ref{fig:comparison_of_decoders_and_distance_bound}, for certain dimensions $k$, there is still room for improvement in Robust Repair Scheme 2 compared to the number of errors guaranteed by the character sum bound. Therefore, we further analyze the error-correction capabilities by replacing the Berlekamp--Welch algorithm with the Guruswami--Sudan list decoding algorithm. Uniqueness of the valid output is guaranteed by capping the decoding radius at $\lfloor(d_{\text{CS}}-1)/2\rfloor$.

The remainder of this paper is organized as follows. In Section~\ref{sec:prelim}, we review the necessary preliminaries regarding Reed-Solomon codes, cyclic codes, and the Guruswami-Wootters trace repair framework, and formally state our problem. In Section~\ref{sec:support_of_rt}, we analyze the support of the repair-trace polynomials, establishing its relation to cyclotomic cosets, and introduce the cyclic supercode $\cT(\FF^*, k)$ of the repair-trace code. In Section~\ref{sec:dim_up_bound}, we introduce a greedy pruning algorithm to obtain an upper bound for the dimension that allows us to correct a given number of errors, and derive an explicit and optimal dimension upper bound for the binary case to correct at least one error. In Section~\ref{sec:lowerbounds}, we establish three theoretical lower bounds for the minimum distance of the repair-trace code, namely the BCH bound, the Degree bound, and the Character Sum bound. Finally, in Section~\ref{sec:decoders}, we propose two robust repair schemes. We present Robust Repair Scheme 1, which efficiently achieves the BCH bound, followed by Robust Repair Scheme 2 and its list-decoding extension, which systematically reduce the polynomial support to improve the error-correction capability closer to the theoretical guarantee by the Character Sum bound.

\begin{figure}
    \centering
    \includegraphics[width=\linewidth]{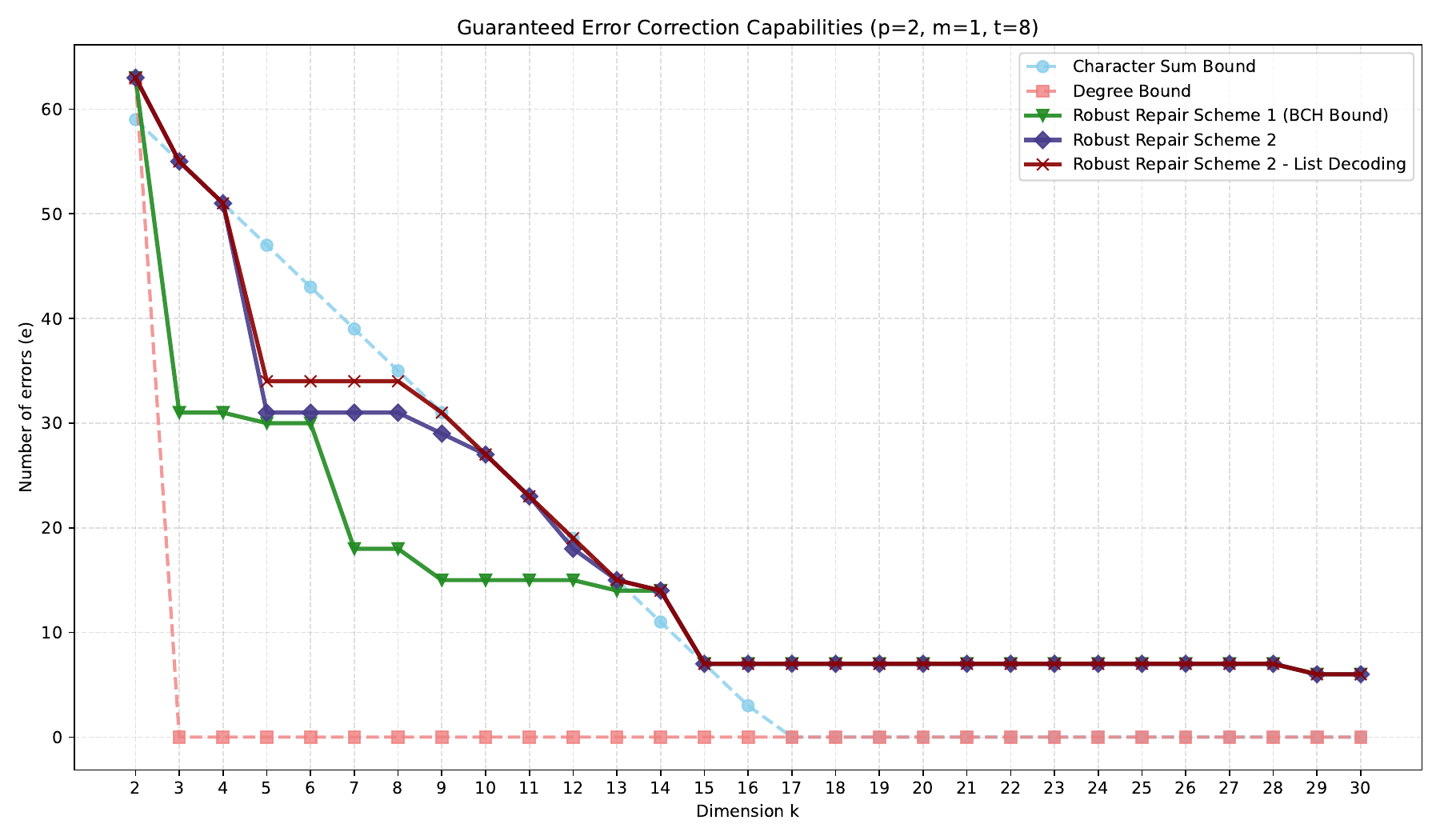}
    \caption{Number of correctable errors guaranteed by established lower bounds of minimum distance of $\cRT(\FF^*,k)$ with $\FF = \GF(2^8)$ and $\BB = \GF(2)$, varying $k$.}
    \label{fig:comparison_of_decoders_and_distance_bound}
\end{figure}

\section{Preliminaries}\label{sec:prelim}

Let $[n]$ denote the set of integers $\{1,\ldots,n\}$, and let $[a,b]$ denote the set of integers from $a$ to $b$. Let $\BB$ be the finite field with $p^m$ elements for some prime $p$ and integer $m\ge 1$, and let $\FF$ be its extension of degree $t\ge 1$. As a result, $|\FF|=p^{mt}$ and $|\BB|=p^m$. Let $\FF^* = \FF\setminus\{0\}$. 
We refer to the elements of $\FF$ as \textit{symbols} and to the elements of $\BB$ as \textit{sub-symbols}. We use $\FF[x]$ to denote the ring of polynomials over the finite field $\FF$.

It is known that an $\FF$-linear $[n, k]$ code $\cC$ is a $k$-dimensional subspace of $\FF^n$. We denote the dual of the code $\cC$ by $\cC^{\perp}$, and clearly, for each $\vc=(c_1,\ldots,c_n)\in\cC$ and $\vc^{\perp}=(c_1^{\perp},\ldots,c_n^{\perp})\in\cC^{\perp}$, it holds that $\sum_{i=1}^n c_i c_i^{\perp}=0$. We denote the minimum distance of $\cC$ by $d(\cC)$. The Singleton bound states that $d(\cC)\leq n-k+1$ (see, for example,~\cite{roth2006}). Codes that attain this bound are called maximum-distance separable (MDS) codes. In this paper, we focus on Reed-Solomon codes, which are MDS. Let us formally define them below.

\begin{definition}\label{RS}

The {\em Reed-Solomon} code $\rs(\cA,k)$ over finite field $\FF$ of dimension $k$ with evaluation points $\cA\subseteq\FF$ is defined as:
\begin{equation*}
\rs(\cA,k) \triangleq \{(f(\alpha))_{\alpha\in\cA}: f\in\FF[x], \deg(f(X))\le k-1\}\,,    
\end{equation*}
while the {\em generalized Reed-Solomon} code $\grs(\cA,k,\vlambda)$ of dimension $k$ with evaluation points $\cA\subseteq\FF$ and multiplier vector $\vlambda\in(\FF\setminus\{\mathbf{0}\})^n$ is defined as:  
{\small
\begin{equation*}
\grs(\cA,k,\vlambda) \triangleq \{(\lambda_\alpha r(\alpha))_{\alpha\in\cA}: f\in\FF[x], \deg(r(X))\le k-1\}\,.
\end{equation*}
}
\end{definition}

Clearly, the generalized Reed-Solomon code with multiplier vector $\vlambda=(1,\ldots,1)$ is \textcolor{black}{a Reed-Solomon code} of the same length and dimension. It is well known (see~\cite{roth2006}) that dual of $\rs(\cA,k)$ is $\grs(\cA,|\cA|-k,\vlambda)$ for $\vlambda=(\lambda_{\alpha})_{\alpha\in \cA}$ where
\begin{equation}
    \lambda_{\alpha_j}=\frac{1}{\prod_{\alpha_i\in\cA\setminus\{\alpha_j\}}(\alpha_j-\alpha_i)}.
\end{equation}

Note that when $\cA=\FF$, we have $\lambda_{\alpha}=1$ for all $\alpha\in\cA$. If it is clear from the context, we use $f(x)$ to denote the polynomial of degree at most $k-1$ corresponding to $\rs(\cA,k)$ and $r(x)$ to denote the polynomial of degree at most $|\cA|-k-1$ corresponding to the dual codeword in $\cC^{\perp}$.

A linear code ${\cal C}$ of length $n$ is {\em cyclic} if $(c_{n-1},c_0,\ldots,c_{n-2})\in {\cal C}$ whenever $(c_0,c_1,\ldots,c_{n-1})\in{\cal C}$. 
Typically, we express a codeword $(c_0,c_1,\ldots,c_{n-1})$ of cyclic code as a polynomial $c(x) = \sum_{i=0}^{n-1} c_ix^i$. Using this polynomial representation, it can be shown that if $\cC$ has dimension $k$, then
\begin{align}
        {\cal C} = \{c(x) = m(x)g(x) : \deg(m(x))\le k-1 \}
\end{align}
\noindent for some {\em generator polynomial }$g(x)$ of degree $n-k$ such that $g(x)$ divides $x^n-1$. The set of roots of $g(x)$ is called {\em the defining set of zeros} of ${\cal C}$. 

For cyclic codes, the structure of the defining set of zeros can be used to determine a lower bound on the minimum distance of the code. Examples include the famous BCH bound (see \cite{MacWilliams1977}) and its improvements by Hartmann-Tzeng~\cite{hartmanntzeng1972} and Roos~\cite{roos1982}. Let us formally state the BCH bound in Theorem~\ref{thm:bchbound}, as our schemes are inspired by it.

\begin{theorem}[BCH Bound~{\cite{MacWilliams1977}}]\label{thm:bchbound}
    Let $T$ be a defining set of zeros of a cyclic code $\cC$. If $T^*(c_1,c_2) \triangleq \{\alpha^{c_1+ic_2}: i\in[1,\delta-1]\}\subseteq T$ and $\gcd(n,c_2) = 1$, then the minimum distance of $\cC$ is at least $\delta$.
\end{theorem}

Our further derivations rely on the fact that a Reed-Solomon code over $\FF$ with evaluation points $\cA = \FF^*$ is cyclic, as formally stated in the following proposition.

\begin{proposition}[{\cite{MacWilliams1977}}]
    Let $\cA = \FF\setminus\{0\}$ and $n = |\cA|$. The Reed-Solomon code $\rs(\cA, k)$ is a cyclic code with generator polynomial {$g(x) = \prod_{\gamma\in T} (x-\gamma)$}, where $T = \{\alpha^i:i\in[1,n-k]\}$.
\end{proposition}

\subsection{Guruswami-Wootters Trace Repair Framework}\label{sec::tracerepairframework}

In this section, we formally introduce the trace repair framework established by Guruswami and Wootters \cite{guruswamiwooters2017} to repair a single erasure in a {full-length} Reed-Solomon code. The main idea of the trace repair framework is to recover a symbol in $\FF$ by using sufficiently many sub-symbols in $\BB$. Specifically, suppose that $f(\alpha^*)$ is erased and let $\cA = \FF\setminus\{\alpha^*\}$ be the set of available evaluation points. 

In what follows, we consider the trace function $\tr:\FF\to\BB$ defined as
\begin{align*}
    \tr(x) = \sum_{i=0}^{t-1} {x^{|\BB|}}^i,\quad\text{for all }x\in\FF.
\end{align*}

Clearly, $\tr(x)$ is a polynomial in $x$ with degree $p^{mt-m}$. In what follows, we treat $\FF$ as a $\BB$-linear vector space of dimension $t$ and let $\{u_1,\ldots,u_t\}$ be a basis of $\FF$ over $\BB$. It is well known that there exists a {\em trace-dual basis} $\{\widetilde{u}_1,\ldots,\widetilde{u}_t\}$ for $\FF$ such that $\tr(u_i\widetilde{u}_j)=1$ if $i=j$, and $\tr(u_i\widetilde{u}_j)=0$ otherwise. The following result plays a crucial role in the repair framework.

\begin{proposition}[{\cite[Ch. 2]{Lidl1996}}]\label{prop:trace}
	Let $\{u_1,\ldots,u_t\}$ be a $\BB$-basis of $\mathbb{F}$. Then there exists a trace-dual basis  $\{\widetilde{u}_1,\ldots,\widetilde{u}_t\}$ and we can write each element $x\in\mathbb{F}$ as
	\begin{equation*}
		x=\sum_{i=1}^t\tr(u_ix)\widetilde{u}_i.
	\end{equation*}
\end{proposition}

This means that to recover the value of $f(\alpha^*)$, we need to determine $\tr(u_if(\alpha^*))$ for all $i=1,\ldots,t$ using information downloaded from the available helper nodes. To achieve this, let
\begin{equation}
    r_i(x) = \frac{\tr(u_i(x-\alpha^*))}{x-\alpha^*},\quad\text{for all }i=1,\ldots,t.\label{eq:trace-definition}
\end{equation}

We can easily verify that $r_i$ is a polynomial of degree $p^{mt-m}-1$ and that $r_i(\alpha^*) = u_i$. If $k \le p^{mt} - p^{mt-m}$, then $r_i$ belongs to the dual code, and the following parity-check equations hold:
\begin{align}\label{eq:paritycheckGW}
    u_if(\alpha^*) = -\sum_{\alpha\in\FF\setminus\{\alpha^*\}} r_i(\alpha)f(\alpha).
\end{align}

Applying the trace function to both sides of \eqref{eq:paritycheckGW} and utilizing its $\BB$-linearity, along with the property that for $a \in \BB$, it holds that $\tr(a\cdot\alpha) = a\tr(\alpha)$, we obtain:
\begin{align*}
    \tr(u_if(\alpha^*)) &= -\sum_{\alpha\in\FF\setminus\{\alpha^*\}} \tr(r_i(\alpha)f(\alpha))\notag\\
    &= -\sum_{\alpha\in\FF\setminus\{\alpha^*\}} \tr(u_i(\alpha-\alpha^*))\tr\left(\frac{ f(\alpha)}{\alpha-\alpha^*}\right). \label{eq:trace-repair}
\end{align*}

Therefore, it suffices to download {\color{black}$\tr( f(\alpha)/(\alpha-\alpha^*))$} from each node $\alpha \in \FF\setminus\{\alpha^*\}$ to obtain $\tr(u_i f(\alpha^*))$ and, consequently, $f(\alpha^*)$. This motivates us to study the code formed by the downloaded traces, which we formally define below:
\begin{definition}\label{def:repairtracecode}
   For a fixed $\alpha^*\in\FF$, the {\em repair-trace} code with evaluation points $\cA = \FF\setminus\{\alpha^*\}$ is defined as:
   
\begin{equation*}
    \cRT(\cA,k) \triangleq \left\{ \left(\tr\left(\frac{f(\alpha)}{\alpha-\alpha^*}\right)\right)_{\alpha\in\cA} : \begin{aligned} &f\in\FF[x], \\ \deg(&f(x))\le k-1 \end{aligned} \right\} \,.
\end{equation*}
\end{definition}
\noindent The above techniques to repair of a single erased symbol in a Reed-Solomon code can be summarized in Theorem~\ref{thm:GW}.

\begin{theorem}[{\cite[\textcolor{black}{Guruswami-Wootters}]{guruswamiwooters2017}}]\label{thm:GW}
    Let $f(\alpha^*)$ be the erased code symbol. If $k\le p^{mt} - p^{mt-m}$ and given $\vc \in \cT(\cA,k)$, we can efficiently compute $f(\alpha^*)$.
\end{theorem}

\subsection{Problem Statement and Reduction to $f(0)$}\label{sec:prob_statement_red_to_f0}

In a practical distributed storage environment, we assume that up to $e$ helper nodes may provide erroneous responses. Our primary objective is to determine the theoretical limits of $k$ and $e$ that guarantee successful recovery of $f(\alpha^*)$ through the trace-mapping framework, and to design efficient decoders that achieve these limits.

In this work, we focus strictly on the scenario of a full-length Reed-Solomon code, where traces are downloaded from all other nodes in the field, meaning $\mathcal{A} = \FF \setminus \{\alpha^*\}$. Extending to non-full-length Reed-Solomon codes where $\mathcal{A} \subset \FF \setminus \{\alpha^*\}$ is highly non-trivial and we defer it to the future work.

To simplify the analysis of the polynomial support, we observe that the repair problem for any arbitrary node $\alpha^*$ can be reduced to the repair of the node at evaluation point $0$. By employing a simple change of variables $u = x - \alpha^*$, the evaluations over $x \in \mathbb{F} \setminus \{\alpha^*\}$ map bijectively to evaluations over $u \in \FF^*$. Consequently, the downloaded traces can be rewritten as
\begin{equation*}
\tr\left(\frac{f(u+\alpha^*)}{u}\right) \quad \text{for } u \in \FF^*.
\end{equation*}

Now, let us define $g(u) \triangleq f(u+\alpha^*)$. Since the binomial expansion of $(u+\alpha^*)^i$ preserves the maximum degree, $g(u)$ is also a polynomial in $\mathbb{F}[u]$ of degree at most $k-1$. Furthermore, the mapping $f(x) \mapsto g(x) = f(x+\alpha^*)$ is an automorphism on the vector space of polynomials of degree at most $k-1$, the code $\cRT(\FF \setminus \{\alpha^*\}, k)$ is equivalent to the code $\cRT(\FF^*, k)$, up to a permutation of coordinates. Hence, without loss of generality, we set $\alpha^* = 0$ for the remainder of this paper. We focus on analyzing the code $\cRT(\FF^*, k)$ corresponding to the trace polynomial $F(x) = \tr(f(x)/x)$ used to repair the value $f(0)$. 

\section{The Support of the Repair-Trace Codes}\label{sec:support_of_rt}

In this section, we analyze the position of nonzero coefficients of $F(x) = \tr(f(x)/x)$ which we refer to as the \textit{support} of $F$, denoted by $\supp(F)$. To this end, we recall the cyclotomic coset definition.

\begin{definition}
    Fix $t$ and prime power $q$. A subset $\{a_1,\ldots,a_s\} \subset \{0,1,\ldots,q^t-2\}$ is called a \textit{cyclotomic coset} if $qa_j = a_{j+1}(\text{mod } q^t-1)$ for all $j\in\{1,\ldots, s-1\}$ and $qa_s = a_1(\text{mod } q^t-1)$. 
    The collection of all such cosets partitions $\{0,1,\ldots,q^t-2\}$ and we refer to it as the \textit{collection of cyclotomic cosets modulo $q^t-1$}.
\end{definition}

We call the smallest element in a cyclotomic coset as the \textit{representative} and we use $C_{r_i}$ to denote the cyclotomic coset with the representative $r_i$. We write the collection of cyclotomic cosets $\Xi$ in an increasing order of its representative, that is,
$$
\Xi = \{C_{r_1}, C_{r_2}, C_{r_3},\ldots, C_{r_{|\Xi|}}\},
$$
where $r_1<r_2<\cdots < r_{|\Xi|}$. 

Observe that the trace $\tr(\cdot)$ is $\BB$-linear, and the expansion
$\tr(f_{i+1}x^i) = \sum_{j=0}^{t-1} f_{i+1} x^{iq^{j}}$
has exponents $\{iq^j \bmod (q^t-1) : 0 \le j \le t-1\}$, which form a cyclotomic coset. 
Hence, the support of $F$ is a union of cyclotomic cosets, as we will further prove formally in Theorem~\ref{thm:sup_is_union_of_cyc_coset}. 
Moreover, the support is formed by including cosets in increasing order of their representatives up to a certain point, and always contains the coset with the largest representative, $C_{r_{|\Xi|}}$. 
We therefore recall the closed-form expression of $r_{|\Xi|}$ in Lemma~\ref{lem:largest_rep}, established by Ding \textit{et al.}~\cite{ding17}. 
While explicit expressions for the next largest representatives were also derived in~\cite{ding17}, we provide an alternative derivation in the binary case ($q=2$) based on the binary representation of $r_{|\Xi|}$.

\begin{lemma}[{Ding et al. \cite{ding17}}]\label{lem:largest_rep}
    Fix $t$ and prime power $q$. Then, $r_{|\Xi|} = q^t - q^{t-1} - 1$.
\end{lemma}

\begin{theorem}
\label{thm:sup_is_union_of_cyc_coset}
     Let $f(x)$ be a polynomial of degree at most $k-1$ in $\FF[x]$. Set $F(x) \triangleq \tr(f(x)/x)$, then its support is given a union of cyclotomic cosets. Specifically, we have that 
     $$\mathrm{supp}(F) \subseteq {\mathcal S}_k \triangleq C_{r_{|\Xi|}}\cup\left(\bigcup_{r\le k-2} C_{r}\right).$$
\end{theorem}

\begin{proof}
     We expand the polynomial $F$,
    \begin{align*}
        F(x) &= \tr\left(f_0/x + \sum_{i\in[0,k-2]} f_{i+1} x^{i}\right)\\
        &= \tr(f_0 x^{q^t-2}) + \sum_{i\in[0,k-2]} \tr(f_{i+1} x^i).
    \end{align*}
    
    The support of $F$ is the union of the supports of each traces $\tr(f_0 x^{q^t-2})$ and $\tr(f_{i+1} x^i)$ for all $i\in[0,k-2]$. Furthermore, the expansion of $\tr(f_{i+1}x^i) = \sum_{j=0}^{t-1} f_{i+1} x^{iq^{j}}$ has exponents $\{iq^j\mod(q^t-1)|j\in[0,t-1]\}$, which is a cyclotomic coset containing $i$. Then, for each $i\in[0,k-2]$, the cyclotomic coset $C_r$ containing $i$ must satisfy $r\le i \le k-2$. Hence, $$\mathrm{supp}\left(\sum_{i\in[0,k-2]} \tr(f_{i+1} x^i)\right)\subseteq \bigcup_{r\le k-2} C_{r}.$$
    
    We are left to show that $\mathrm{supp}(\tr(f_0x^{q^t-2})) \subseteq C_{r_{|\Xi|}}$, i.e., the cyclotomic coset containing $q^t-2$ is the one containing $r_{|\Xi|}$, the largest representative. Let $C'$ be the cyclotomic coset containing $q^t - 2$, then
    $$
    C' = \{q^t - 2, q^t - 1 - q, q^t - 1 - q^2, \ldots, q^t - 1 - q^{t-1}\}.
    $$
    Clearly, the representative of $C'$ is $q^t - q^{t-1} - 1$, which by Lemma~\ref{lem:largest_rep} is the largest representative.
\end{proof}


\begin{example}\label{ex:support_of_F}
Fix $\FF = \mathrm{GF}(2^5), \BB = \mathrm{GF}(2)$. Then, we have
\begin{align*}
    C_0 &= \{0\}\\
    C_1 &= \{1,2,4,8,16\}\\
    C_3 &= \{3,6,12,24,17\}\\
    C_5 &= \{5,10,20,9,18\}\\
    C_7 &= \{7,14,28,25,19\}\\
    C_{11} &= \{11,22,13,26,21\}\\
    C_{15} &= \{15,30,29,27,23\}
\end{align*}
Here, $r_1 = 0 < r_2 = 1 < r_3 = 3 < \cdots < r_7 = 15$. Suppose that $f$ is a polynomial of degree at most $9$, then 
$$f(x)/x = f_0/x + \sum_{i\in[1,9]}f_i x^{i-1} = \left(\sum_{i\in[1,9]}f_i x^{i-1}\right) + f_0 x^{30} ,$$
\begin{align*}
F(x) &= \tr(f(x)/x) \\
&= \tr \Biggl( f_0 x^{30} + f_1 + \sum_{i\in\{1,2,4,8\}}f_{i+1}x^i \\
&\qquad\qquad\qquad\quad\,\,\,\, + \sum_{i\in\{3,6\}}f_{i+1} x^i + f_6 x^5 + f_8 x^7 \Biggr) \\
&= \sum_{i\in C_{r_7} \cup \big(\bigcup_{j=1}^5 C_{r_j}\big)} F_i x^i.
\end{align*}
In other words, the support of $F(x) = \tr(f(x)/x)$ is $C_{r_7}\cup \left(\bigcup_{j\in[1,5]} C_{r_j}\right)$, where each $r_j \le 8$.
\end{example}

\subsection{Cyclic Supercode $\cT(\FF^*,k)$ of Repair-Trace Code $\cRT(\FF^*,k)$}

Motivated by Theorem~\ref{thm:sup_is_union_of_cyc_coset}, we can embed the repair-trace code into a $\FF$-linear supercode. Specifically, we consider the supercode $\cT(\FF^*,k)$, where
\begin{align}\label{eq:supercode}
    \cT(\FF^*,k) = \{(F(\alpha))_{\alpha\in\FF^*}: \supp(F) \subseteq \cS_k\}. 
\end{align}

In fact, the repair-trace code $\cRT(\FF^*,k)$ is a subfield subcode of $\cT(\FF^*,k)$, i.e.,
$\cRT(\FF^*,k) = \cT(\FF^*,k) \cap \BB^n$.
In \cite{kim2024decoding}, we established that any sparse Reed-Solomon code over $\FF$ with evaluation points $\cA = \FF^*$ is a cyclic code. Specializing this property to $\cT(\FF^*,k)$ yields Theorem~\ref{thm:rep_trace_is_cyc}.

\begin{theorem}\label{thm:rep_trace_is_cyc}
    The code $\cT(\FF^*,k)$ is a cyclic code with generator polynomial $g(x) = \prod_{\gamma\in T}(x-\gamma)$, where $T = \{\alpha^{n-i}: i \in \{0,1,\ldots,n-1\}\setminus {\mathcal{S}}_k\}$.
\end{theorem}

Theorem~\ref{thm:rep_trace_is_cyc} shows a strong connection between $\cS_k$ and the defining set of zeros of $\cT(\FF^*,k)$. Specifically, if the term $x^c$ is not in $F(x)$, i.e., $c\notin \cS_k$, then $\alpha^{n-c}\in T$, the defining set of zeros of $\cT(\FF^*,k)$.
Hence, the BCH bound guarantees a minimum distance of $\delta+1$ and error correction up to $\lfloor\delta/2\rfloor$ errors, if there are $\delta$ consecutive zero coefficients in $F$.

\section{Dimension Upper Bounds for Error Correction}\label{sec:dim_up_bound}

In this section, we provide an algorithm to obtain an upper bound for the dimension $k$ of a full-length Reed-Solomon code $\rs(\mathbb{F},k)$ such that a single erased node can be repaired by using the trace-repair scheme despite the presence of $e$ erroneous traces. 
To correct $e$ erroneous traces, the repair-trace code must possess a minimum distance of at least $2e+1$. Based on our previous observations, this distance is guaranteed if we can find an integer $b$ coprime to $n$ such that the permuted polynomial $G(x) = F(x^b)$ has a sequence of $2e$ consecutive zero coefficients. 

The zero coefficients of $F(x)$ correspond exactly to the cyclotomic cosets that are \textit{excluded} from the support $\mathcal{S}_k$. Recall from Theorem~1 that the support $\mathcal{S}_k$ always contains the coset $C_{r_{|\Xi|}}$ and is constructed by including cosets in increasing order of their representatives, starting from $r_1$. Therefore, the zero coefficients of $F(x)$ will naturally correspond to the omitted cosets. Since the support is built by including cosets with the smallest representatives first, the omitted cosets are strictly those with the large representatives (excluding the largest).
Because our goal is to find an \textit{upper bound} for the dimension $k$, we must maximize the size of the support $\mathcal{S}_k$. This implies we should exclude the minimum possible number of cosets needed such that the excluded cosets contain $2e$ consecutive terms. This naturally leads to a greedy pruning approach:
\begin{enumerate}
    \item We start by excluding only the second largest available coset, $C_{r_{|\Xi|-1}}$, and check if there exists any integer $b$ that is coprime to $n$
    that organizes these elements into $2e$ consecutive integers modulo $n$.
    \item If no such $b$ exists, the current set of zeros is insufficient. We then prune the next largest coset, $C_{r_{|\Xi|-2}}$, adding its elements to our pool of zeros, and test all permutations again.
    \item We iterate this process, pruning cosets one by one from the top down, stopping when we find a valid permutation $b$ that successfully forms the $2e$ consecutive terms.
\end{enumerate}
The index at which this process terminates directly dictates the optimal maximum dimension $k$. This procedure is formalized in Algorithm~K.

\begin{figure}[htbp]
    \centering
    \vspace{0cm}
    \fbox{
        \begin{minipage}{0.5\textwidth}
        \vspace{0.5em}
        \noindent \textbf{Algorithm K} (Greedy Pruning Algorithm for Dimension Upper Bound)

        \vspace{0.5em}
        \noindent\textbf{Input:} Erroneous traces correction capability $e$, code length $n = p^{mt}-1$, and the cyclotomic cosets $C_{r_1},\ldots, C_{r_{|\Xi|}}$.

        \noindent\textbf{Output: } Dimension upper bound $K_{\text{max}}$, or a failure flag.

        \begin{enumerate}[label = \textbf{K\arabic*.}, leftmargin=*, labelsep = 0.5em]
            \item \textbf{(Initialization)} Set $z \gets 0$, the removed cosets $\mathcal{R} \gets \emptyset$, \texttt{SeqFound} $\gets$ \texttt{False}.
            \item \textbf{(Pruning Iteration)} While \texttt{SeqFound} is \texttt{False} and $z < |\Xi|$:
            \begin{enumerate}[label = \textbf{K2.\arabic*.}, leftmargin = 0.7em, labelsep = 0.5em]
                \item Update $z \gets z + 1$ and $\mathcal{R} \gets \mathcal{R} \cup C_{r_{|\Xi|-z}}$.
                \item For each integer $b\in \{1,\ldots,n-1\}$ where $\gcd(b,n) = 1$:
                \begin{enumerate}[label = \textbf{K2.2.\arabic*.}, leftmargin =0.7em, labelsep = 0.5em]
                    \item Compute the permuted set $\cR^{(b)} \gets \{bx \pmod{n}:x\in\mathcal{R}\}$.
                    \item If $\cR^{(b)}$ contains a sequence of $2e$ consecutive integers modulo $n$, update \texttt{SeqFound} $\gets$ \texttt{True}.
                \end{enumerate}
            \end{enumerate}
            \item \textbf{(Output)} If \texttt{SeqFound} is \texttt{True}, return the upper bound $K_{\text{max}} \gets r_{|\Xi|-z} + 1$. Otherwise, return ``Not achievable''.
        \end{enumerate}
        \vspace{0.5em}
        \end{minipage}
    }
    \label{fig:greedy_pruning}
\end{figure}

\begin{figure}[htbp]
    \centering
    \begin{subfigure}[b]{0.5\textwidth}
        \centering
        \includegraphics[width=0.75\textwidth]{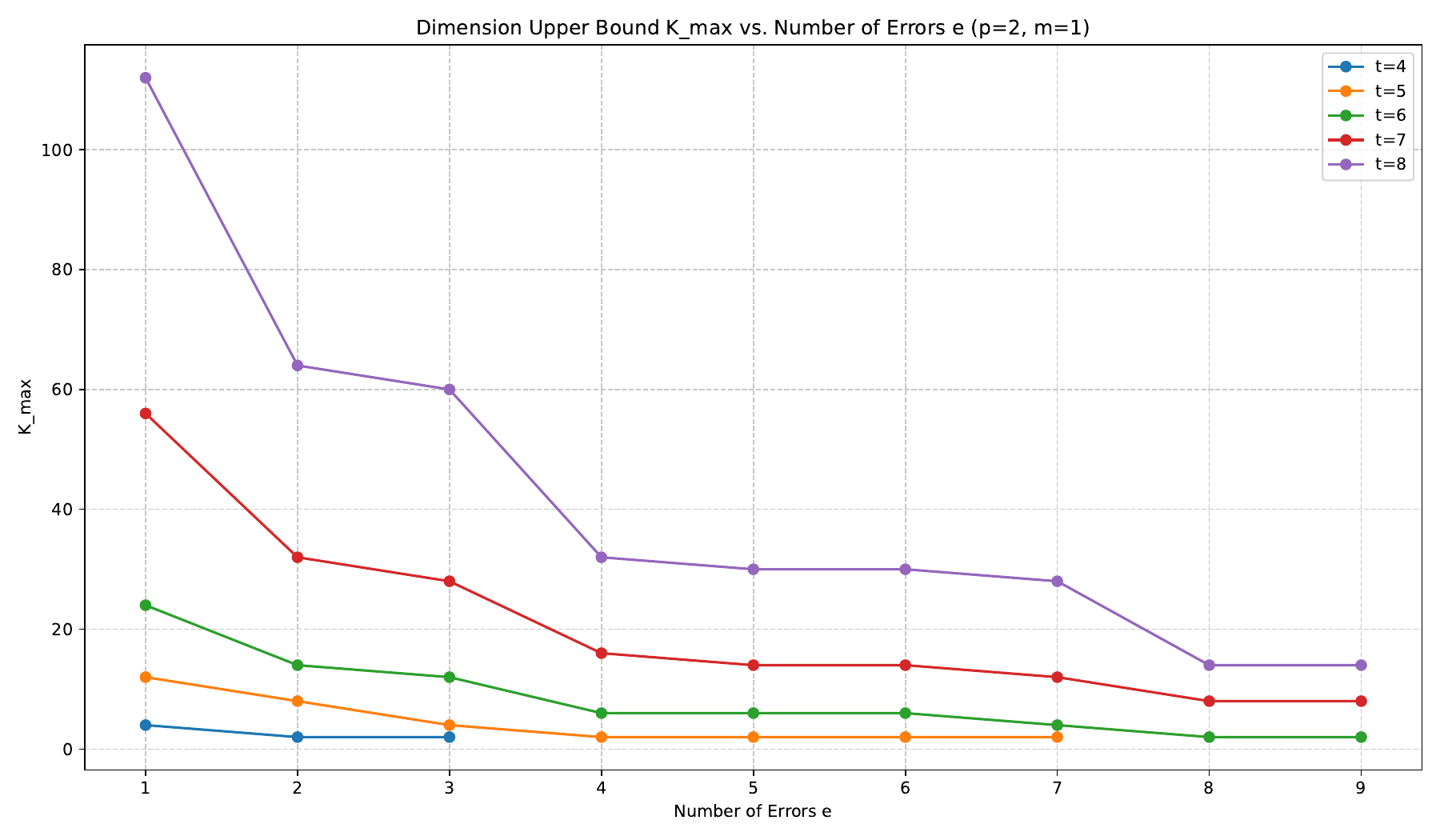}
        \caption{Dimension Upper Bound $K_{\max}$}
        \label{fig:upper_bound_plot}
    \end{subfigure}
    \hfill 
    \begin{subfigure}[b]{0.5\textwidth}
        \centering
        \includegraphics[width=0.75\textwidth]{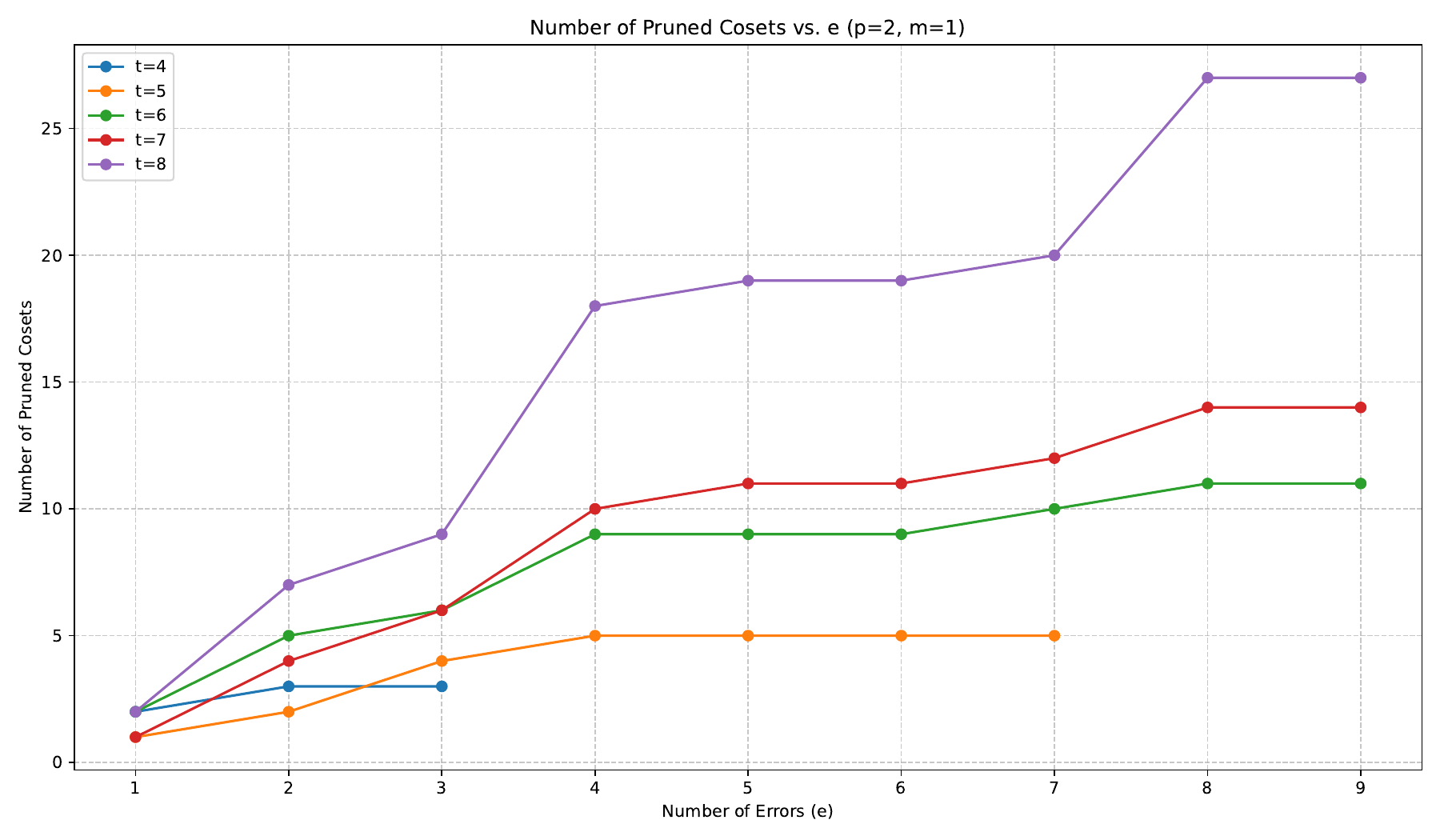}
        \caption{Number of Pruned Cosets}
        \label{fig:pruned_cosets_plot}
    \end{subfigure}
    
    \caption{The dimension upper bound $K_{\max}$ (left) and the corresponding number of pruned cosets required (right) versus the error-correction capability $e$, for $p=2, m=1$, varying $t$.}
    \label{fig:kmax_and_pruning_analysis}
\end{figure}

\subsection{Binary Case}

While Algorithm~K computes an upper bound on the dimension $k$ for general fields and error-correction capability $e$, obtaining a closed-form expression remains difficult. 
Nevertheless, for the case $q=2$, we leverage on the structure of the cyclotomic cosets to analytically derive the exact maximum dimension required to guarantee the single error-correction ($e=1$).
In particular, the bound given by Algorithm~K is determined by either the second or third largest representative. Moreover, this upper bound is sharp;
see Theorem~\ref{thm:binary_bound}.

Before that, we need the closed-form expression for coset representatives for the binary case $q = 2$, presented in Lemma~\ref{lem:binary_rep}. The general forms of these representatives for any prime power $q$ were established by Ding et al.~\cite{ding17}. In this work, we provide an alternative proof focused on the binary case in the Appendix.

\begin{lemma}[{Ding et al. \cite{ding17}}]\label{lem:binary_rep}
        If $q = 2$, then 
    \begin{enumerate}
        \item $r_{|\Xi|} = 2^{t-1} - 1$
        \item $r_{|\Xi|-1} = 2^{t-1} - 2^{\lfloor(t-1)/2\rfloor} - 1$, for $t\ge 3$
        \item $r_{|\Xi|-2} = 2^{t-1} - 2^{\lfloor(t-1)/2\rfloor+ 1} - 1$, for $t\ge 4$
    \end{enumerate} 
\end{lemma}

With the explicit closed-form expressions for the representatives established, we can now analyze the structure of the largest cyclotomic cosets. The following theorem demonstrates that these pruned cosets, specifically $C_{r_{|\Xi|-1}}$ for odd $t$ and $C_{r_{|\Xi|-1}}\cup C_{r_{|\Xi|-2}}$ for even $t$, contains two consecutive terms.
\begin{theorem}\label{thm:binary_consecutive}
    If $q = 2$ and $t\ge 3$ is odd, then the cyclotomic coset $C_{r_{|\Xi|-1}}$ contains two consecutive terms, specifically, $2r_{|\Xi|-1}$ and $2r_{|\Xi|-1}-1$. If $q=2$ and $t \ge 4$ is even, then the cyclotomic cosets $C_{r_{|\Xi|-1}}$ and $C_{r_{|\Xi|-2}}$ contains two consecutive terms, specifically, $2^2r_{|\Xi|-1}$ and $2 {r_{|\Xi|-2}}$.
\end{theorem}
\begin{proof}
    \underline{For odd $t=2k+1$.} The term $2r_{|\Xi|-1}$ is, clearly, in $C_{r_{|\Xi|-1}}$ and we are left to show that $2r_{|\Xi|-1}-1$ is also in it. From Lemma~\ref{lem:binary_rep}, we know that $r_{|\Xi|-1}$ has the binary representation $0\boldsymbol{1}_{k-1}0 \boldsymbol{1}_{k}$. Then,
    \begin{itemize}
        \item $2r_{|\Xi|-1}$ has the binary representation $\boldsymbol{1}_{k-1}0 \boldsymbol{1}_{k-1}10$, and
        \item $2r_{|\Xi|-1}-1$ has the binary representation $\boldsymbol{1}_{k-1}0 \boldsymbol{1}_{k-1}01$.
    \end{itemize}
    Clearly, the binary representation $\boldsymbol{1}_{k-1}0 \boldsymbol{1}_{k-1}01$ can be obtained by rotating $\boldsymbol{1}_{k-1}0 \boldsymbol{1}_{k}0$, $k+1$ times. In other words, $2r_{|\Xi|-1}-1 = 2^{k+1} (2r_{|\Xi|-1})\in C_{r_{|\Xi|-1}}$.

    \underline{For even $t=2k$.} The terms $2^2r_{|\Xi|-1}$ and $2 {r_{|\Xi|-2}} $ are both in $C_{r_{|\Xi|-1}}$ and $C_{r_{|\Xi|-2}}$, respectively. For even $t = 2k$, from Lemma~\ref{lem:binary_rep}, $r_{|\Xi|-1}$ and $r_{|\Xi|-2}$ has the binary representation $0\boldsymbol{1}_{k-1} 0\boldsymbol{1}_{k-1}$ and $0\boldsymbol{1}_{k-2} 0 \boldsymbol{1}_{k}$, respectively. Then,
    \begin{itemize}
        \item $2^2r_{|\Xi|-1}$ has the binary representation $\boldsymbol{1}_{k-2} 0\boldsymbol{1}_{k-1}01$, and
        \item $2r_{|\Xi|-2}$ has the binary representation $\boldsymbol{1}_{k-2}0\boldsymbol{1}_{k-1}1 0$,
    \end{itemize}
    which are two consecutive terms.
\end{proof}

Theorem~\ref{thm:binary_consecutive} establishes the existence of a sequence of two consecutive zero coefficients in the permuted repair-trace polynomial when the appropriate cosets are pruned. By the BCH bound (Theorem~\ref{thm:bchbound}), this structural gap guarantees a minimum distance of at least 3, allowing for the correction of a single error. We can now formally prove Theorem~\ref{thm:binary_bound}.

\begin{proof}[Proof of Theorem~\ref{thm:binary_bound}]
    The upper-bound $K$ is a direct consequence of Theorem~\ref{thm:binary_consecutive} and the fact that $\cRT(\FF^*,k)\subset \cT(\FF^*,k)$. Now, let's prove that it is optimal. For odd $t\ge 3$, take $$f(x) = \sum_{i\in [1,|\Xi|]} x^{1+r_i}.$$ Then,
    \begin{align*}
    F(x) &= \tr(f(x)/x) =  \tr\left(\sum_{i\in[1,|\Xi|]}x^{r_i}\right)\\ &= \sum_{i\in[0,2^t-2]} x^i = \begin{cases}
        1&\text{if }x=1,\\
        0&\text{if }x\ne 1.
    \end{cases}
    \end{align*}
    
    \noindent For even $t\ge 4$, take $$f(x) = \sum_{i\in[1,|\Xi|]}\left(1+\omega^{-r_i(2^{t/2}-1)}\right)x^{1+r_i}.$$ Then,
    \begin{align*}
    F(x) &= \tr(f(x)/x) = \tr\left(\sum_{i\in[1,|\Xi|]}\left(1+\omega^{-r_i(2^{t/2}-1)}\right)x^{r_i}\right)\\
    &= \sum_{j\in [0,2^t-2]} \left(1+\omega^{-j(2^{t/2}-1)}\right)x^j\\
    &= \sum_{j\in[0,2^t-2]} x^j + \sum_{j\in[0,2^t-2]} \left(\omega^{-(2^{t/2}-1)}x\right)^j\\
    &= \begin{cases}
        1&\text{if } x = 1 \text{ or } x = \omega^{2^{t/2}-1},\\
        0&\text{otherwise}.
    \end{cases}
    \end{align*}
    Furthermore, since $\omega^{-r_{|\Xi|-1} (2^{t/2}-1)} = \omega^{2^{-1}(2^{t}-1)}=1$, we have $\supp(F) = [0,2^t-2]\setminus C_{|\Xi|-1}$. Therefore, $F(x)$ is a corresponding polynomial of $\cRT(\FF^*,k)$.
\end{proof}
\section{Distance Lower Bounds}\label{sec:lowerbounds}

\subsection{BCH Bound}
\label{subsec:bch_bound}

Recall that the repair-trace code $\cRT(\FF^*,k)$ is a subfield subcode of $\cT(\FF^*,k)$. Therefore, the BCH bound of $\cT(\FF^*,k)$ serves as a lower bound for the minimum distance of $\cRT(\FF^*,k)$. This relationship naturally yields the lower bound given in Theorem~\ref{thm:distance_lower_bound}(i). Hence, in this section, we focus on obtaining the BCH bound of $\cT(\FF^*,k)$.
As established in Theorem~\ref{thm:rep_trace_is_cyc}, the code $\mathcal{T}(\mathbb{F}^*, k)$ is a cyclic code of length $n = q^t - 1$, and its defining set of zeros is governed by the zero coefficients of the repair-trace polynomial $F(x) = \mathrm{Tr}(f(x)/x)$, which belong to the set $[0,q^t-2]\setminus \cS_k$. We denote this set of zero coefficients as $\mathcal{R}_k \triangleq [0,q^t-2] \setminus \mathcal{S}_k$. 

To obtain the strongest distance guarantee, we search for the longest sequence of zero coefficients spaced by an integer $b$ coprime to $n$. Equivalently, for each $b$ with $\gcd(b,n) = 1$, we consider $\cR_k^{(b)} \triangleq \{br : r\in\cR_k\}$ and find $\delta_b$, the length of the longest sequence of consecutive integers in it. By maximizing $\delta_b$ over all $b$ coprime to $n$, the BCH bound guarantees a minimum distance of at least $\max_{b}\delta_b + 1$. We formalize this procedure in Algorithm~D, which exhaustively searches for the optimal permutation multiplier $b$ over the zero coefficients set $\mathcal{R}_k$ and returns the highest guaranteed minimum distance $d_{\text{BCH}}$.

\begin{figure}[htbp]
    \centering
    \fbox{
        \begin{minipage}{0.46\textwidth}
        \vspace{0.5em}
        \noindent \textbf{Algorithm D} (BCH Distance Lower Bound)

        \vspace{0.5em}
        \noindent\textbf{Input:} Code length $n=q^t-1$, dimension $k$, and the polynomial support $\mathcal{S}_k$.

        \noindent\textbf{Output: } The BCH lower bound on the minimum distance $d_{BCH}$, optimal multiplier $b_{\text{opt}}$, maximum zero-sequence length $\delta_{\text{opt}}$, and sequence start index $s_{\text{opt}}$.

        \begin{enumerate}[label = \textbf{D.\arabic*.}, leftmargin=*, labelsep = 0.5em]
            \item \textbf{(Initialization)} Set $\delta_{\text{opt}} \gets 0$, $b_{\text{opt}} \gets 1$, $s_{\text{opt}} \gets 0$, and the zero coefficients set $\mathcal{R}_k \gets \{0, 1, \dots, n-1\} \setminus \mathcal{S}_k$.
            \item \textbf{(BCH Bound Optimization)} For each integer $b \in \{1, \dots, n-1\}$ such that $\gcd(b,n) = 1$:
            \begin{enumerate}[label = \textbf{D.2.\arabic*.}, leftmargin = 0.7em, labelsep = 0.5em]
                \item Compute the permuted zero-coefficient set $\mathcal{R}_k^{(b)} \gets \{(b \cdot r) \pmod n : r \in \mathcal{R}_k\}$.
                \item Find $\delta_b$, the length of the longest sequence of consecutive integers (modulo $n$) in $\mathcal{R}_k^{(b)}$, and let $s_b$ be its starting index.
                \item If $\delta_b > \delta_{\text{opt}}$, update $\delta_{\text{opt}} \gets \delta_b$, $b_{\text{opt}} \gets b$, $s_{\text{opt}} \gets s_b$.
            \end{enumerate}
            \item \textbf{(Output)} Return $d_{BCH} \gets \delta_{\text{opt}} + 1$, along with parameters $(\delta_{\text{opt}},b_{\text{opt}}, s_{\text{opt}})$.
        \end{enumerate}
        \vspace{0.5em}
        \end{minipage}
    }
    \label{fig:bch_bound}
\end{figure}

Despite having the BCH bound that can be obtained algorithmically in polynomial time, it is highly valuable to establish explicit, closed-form theoretical distance guarantees. In the next subsections, we establish two other explicit lower bounds: the Degree bound and the Character Sum bound. The Degree bound, in fact, is a special case of the BCH bound when $b = 1$ and is determined by finding consecutive zero coefficients starting from a specific position (which may not always be the longest sequence). {\color{black} On the other hand, the Character Sum bound is obtained analytically by expressing the codeword weights using additive characters and bounding their non-degenerate sums via Weil-like estimates over finite fields.} In terms of performance, the Character Sum bound often yields a higher distance guarantee than the BCH bound for small to moderate values of $k$ (see Fig.\ref{fig:comparison_of_decoders_and_distance_bound} for some numerical example). We discuss in more detail and prove these explicit bounds in the next subsections.
\subsection{Degree Bound}
We consider the code $\cT_1 \triangleq \left\{\left(\alpha^{p^{mt-m}} c_{\alpha}\right)_{\alpha\in\FF^*} : \vc \in \cRT(\FF^*,k)\right\}$
and show that it is a subcode of some Reed-Solomon code.

\begin{proposition}\label{prop:grs-1}
Let $\Delta$ be defined in Theorem~\ref{thm:distance_lower_bound}. If $\Delta<|\cA|$, then
$\cT_1\subseteq\rs(\FF^*,\Delta+1)$.
\end{proposition}

\begin{proof}
We note that $\vc^{*}\in\cT_1$ can be represented as 
\begin{align*}
(F(\alpha))_{\alpha\in\cA}&=\left(\alpha^{p^{mt-m}}\tr\left(\frac{f(\alpha)}{\alpha}\right)\right)_{\alpha\in\cA}\\
&=\left(\alpha^{p^{mt-m}}\sum_{i=0}^{t-1}\left(\frac{f(\alpha)}{\alpha}\right)^{p^{mi}}\right)_{\alpha\in\cA}\\
&=\left(\alpha^{p^{mt-m-1}}f(\alpha)+\cdots+f(\alpha)^{p^{mt-m}}\right)_{\alpha\in\cA},
\end{align*}
where $F(x)$ is a polynomial of degree $\max(k+p^{mt-m}-2,\ldots,(k-1)p^{mt-m})$. This fact finishes the proof. 
\end{proof}

Since $\cT_1\subseteq \rs(\FF^*,\Delta+1)$ and $\cT_1$ is equivalent to the repair-trace code $\cRT(\FF^*,k)$ (since $\alpha\ne0$ for all $\alpha\in\FF^*$), the minimum distance of $\cRT(\FF^*,k)$ is at least $|\cA|-\Delta$ and we obtain Theorem~\ref{thm:distance_lower_bound}(ii). 

\subsection{Character Sum Bound}
We prove the Theorem~\ref{thm:distance_lower_bound}(iii) by modifying the proof of \cite[Theorem~5.4]{roth2006} for two cases, $m=1$ and $m>1$. Before we proceed further, let us provide a short overview of character sums and refer the reader to \cite{Lidl1996, pascale2013} for more details.

Let us assume that $\omega=e^{\frac{2i\pi}{p}}$ is a primitive $p$-th root of complex unity. 
It is well known that for any $x\in\BB$ {\color{black}with $m=1$}  it holds that 
\begin{equation}\label{eq:complexroot2}
\sum_{a\in\BB\setminus\{0\}}\omega^{ax} = 
\begin{cases}
  p-1 & \text{if $x=0$} \\
  -1 & \text{otherwise}.
\end{cases}  
\end{equation}

For any element $a$ from $\FF$, we can define an additive character as function $\chi_a(x)=\omega^{{\rm Abs}\tr(ax)}$, where $x\in\FF$ and ${\rm Abs}\tr(\cdot)$ is the trace function from $\FF$ to the finite filed with $p$ element. Character defined by $\chi_0(x)=1$ is called trivial, while all other characters are called non-trivial. The additive character $\chi_1(x)$ is said to be canonical. It is well known that all additive characters of $\FF$ form a group of order $p^{mt}$ isomorhic to the additive group of $\FF$ and the following property holds
\begin{equation}
\chi_{a+b}(x)=\chi_a(x)\chi_b(x).
\end{equation}
The orthogonality relation of additive characters is given by
\begin{equation}\label{eq:complexroot}
\sum_{x\in \FF}\chi_a(x) =
\begin{cases}
0,& \mbox{if $a\neq 0$}\\
p^{mt},& \mbox{if $a=0$}
\end{cases}
\end{equation}

By the same way, for any element $g^j$ from multiplicative group of $\FF$ we can define a multiplicative character as a function $\Psi_j(g^k)=e^{\frac{2i\pi jk}{p^{mt}-1}}$, where $g$ is a fixed primitive element of $\FF$. Character defined by $\Psi_0(x)=1$ is called trivial, while all other characters are called non-trivial. It is well known that all multiplicative characters of $\FF$ form a group of order $p^{mt}-1$ isomorphic to the multiplicative group of $\FF$ and the following property holds
\begin{equation}
\Psi_{ab}(x)=\Psi_a(x)\Psi_b(x).
\end{equation}

Our further derivations rely on the upper bound for absolute value of the following non-degenerate sum
\begin{equation}\label{nondegeneratesum2}
S(\chi_a, \Psi_b; \phi, \varphi)=\sum_{x\in\FF\setminus\mathcal{S}}\chi_a(\phi(x))\Psi_b(\varphi(x)),
\end{equation}
where $\mathcal{S}$ denotes the set of poles of functions $\phi(x)\in\FF[x]$ and $\varphi(x)\in\FF[x]$. Non-degenerate property means that $a\phi(x)\ne h(x)^p-h(x)+c$ and $\varphi\ne ch(x)^{p^{mt}-1}$ for any $h(x)\in\FF[x]$ and $c\in\FF$. It is clear that $a\phi(x)= h(x)^p-h(x)+c$ and $\varphi(x)= ch(x)^{p^{mt}-1}$ imply that $\chi_a(\phi(x))$ and $\Psi_b(\varphi(x))$ are respective constant numbers for each $x\in\FF\setminus\mathcal{S}$. Essentially, we have the  following generalization of Weil estimate proved by Castro and Moreno to the case of rational functions $\phi(x)$ and $\varphi(x)$ in \cite{CaestroMoreno2000} in notations of \cite{Cochrane2006} and \cite{Booker2022}.
\begin{proposition}[{\cite[Lemma~2.1]{Booker2022}}]
Let $\phi(x)$, $\varphi(x)$ be rational functions in $\FF$, $\chi_a$ be non-trivial additive character on $\FF$ and $\Psi_b$ be non-trivial multiplicative character on $\FF$. Let $\mathcal{S}$ be the set of poles of functions $\phi$ and $\varphi$ in $\FF$. Further, let $l$ be the number of distinct zeros and non-infinite poles of $\phi$. Let
$l_1$ be the number of all poles of $\varphi$ and $l_0$ be the sum of their multiplicities. Let $l_2$ be the number of non-infinite poles of $\varphi$ which are zeros or poles of $\phi$.
Then
\begin{align}\label{weilestimate}
|S(\chi_a, \Psi_b; \phi, \varphi)|&=|\sum_{x\in\FF\setminus\mathcal{S}}\chi_a(\phi)\Psi_b(\varphi)|\notag\\
&\leq(l+l_0+l_1-l_2-2)\sqrt{p^{mt}}
\end{align}
\end{proposition}
By setting $\varphi(x)=1$ and $\phi(x)=\frac{f(x)}{x}$ so that $a\phi(x)\ne h(x)^p-h(x)+c$ for any $h(x)\in\FF[x]$ and $c\in\FF$, we receive the following estimate:
\begin{align}\label{weilestimate2}
\left|\sum_{x\in\FF\setminus\{0\}} \chi_a \left(\frac{f(x)}{x}\right)\right| \le (k-1)\sqrt{p^{mt}}.
\end{align}

\begin{proposition}\label{prop:char-sum}
	If $\cA=\FF\setminus\{0\}$ and $m=1$, 
	then every nonzero word in $\cRT(\cA,k)$ has weight at least
	\begin{equation}\label{charsumbound}
		\frac{p-1}{p}\left(|\cA|-(k-1)\sqrt{p^t}\right).    
	\end{equation}
\end{proposition}
\begin{proof}
We distinguish between two cases.

\textit{Case 1}. $f(x)=x(h(x))^{p}-xh(x)+xb$ for some $h\in \mathbb{F}[x]$ and $b\in \mathbb{F}$. In this case,
\begin{align*}
c_j & =\tr\left(\frac{f(\alpha_{j})} {\alpha_{j}}\right) =\tr\left(h(\alpha_{j})^{p}\right)-\tr(h(\alpha_{j}))+\tr(b)\notag\\
&=\tr\left(h(\alpha_{j})\right)^{p}-\tr(h(\alpha_{j}))+\tr(b)=
\tr(b),
\end{align*}

In other words, $\vc$ is a multiple of the all-ones vector.

\textit{Case 2}. $f(x)\ne x(h(x))^{p}-xh(x)+xb$ for any $h\in \mathbb{F}[x]$ and $b\in \mathbb{F}$. In this case we can form the non-degenerate sum and apply an estimate~\eqref{weilestimate}. For $p$-th root of complex unity we can write down that
\begin{align}
\sum_{j=1}^{p^t-1}\left(\sum_{a\in \BB\setminus\{0\}}\omega^{ac_j}\right)
& =(p-1)(p^t-1-\wt(\vc))-\wt(\vc)\notag\\
&=(p-1)(p^{t}-1)-p\wt(\vc),   
\end{align}
where $\wt(\vc)$ is the Hamming weight of the codeword $\mathbf{c}$.

Utilizing the fact that $\omega^{a\textrm{Tr}(\frac{f(x)}{x})}$ for $a\in \mathbb{B}\setminus\{0\}$ is non-trivial additive character $\chi_a(\frac{f(x)}{x})$ we have
\begin{align}
\left|\sum_{a\in \mathbb{B}\setminus\{0\}}\sum_{j=1}^{p^t-1}w^{ac_j}\right|
& =\left|\sum_{a\in \mathbb{B}\setminus\{0\}}\sum_{x\in \mathbb{F}\setminus\{0\}}\chi_a\left(\frac{f(x)}{x}\right)\right|\notag\\
& \leq \sum_{a\in \mathbb{B}\setminus\{0\}}\left|\sum_{x\in \mathbb{F}\setminus\{0\}}\chi_a\left(\frac{f(x)}{x}\right)\right|.
\end{align}

Applying the estimate~\eqref{weilestimate} we have
\begin{equation}
\left|(p-1)(p^{t}-1)-p\wt(\vc)\right| \leq(p-1)(k-1)\sqrt{p^{t}}.
\end{equation}
Combining two cases, we get the proposition statement. 
\end{proof}

\begin{proposition}\label{prop:char-sum-new}
	If $\cA=\FF\setminus\{0\}$ and $m>1$,
	then every nonzero word in $\cRT(\cA,k)$ has weight at least
	\begin{equation}\label{charsumbound2}
		\frac{p^m-p}{p^m}\left(|\cA|-(k-1)\sqrt{p^{mt}}\right)
	\end{equation}
\end{proposition}
\begin{proof}
Let $\vc=\left(\tr\left(\frac{f(\alpha)} {\alpha}\right)\right)_{\alpha\in\cA}$ be a codeword of $\cRT(\cA,k)$. Let $\lambda_1$ be the canonical additive character of $\BB$. By the orthogonality relation of additive characters, we deduce that
\begin{align*}
\wt(\vc)&=p^{mt}-1-\#\left\{\alpha\in\cA:\tr\left(\frac{f(\alpha)} {\alpha}\right)=0\right\}\\
&=p^{mt}-1-\frac{1}{p^m}\sum_{\alpha\in\cA}\sum_{a\in \BB}\lambda_1\left(a \tr\left(\frac{f(\alpha)} {\alpha}\right)\right)\\
&=p^{mt}-1-\frac{1}{p^m}\sum_{a\in \BB} \sum_{\alpha\in\cA}\chi_a\left(\frac{f(\alpha)} {\alpha}\right)\\
&=\frac{(p^{mt}-1)(p^m-1)}{p^m}-\frac{1}{p^m}\sum_{a\in \BB\setminus\{0\}} \sum_{\alpha\in\cA}\chi_a\left(\frac{f(\alpha)} {\alpha}\right)
\end{align*}
From the above equation, we have
\begin{equation}\label{prop:eq1}
\sum_{a\in \BB\setminus\{0\}} \sum_{\alpha\in\cA}\chi_a\left(\frac{f(\alpha)} {\alpha}\right)=(p^{mt}-1)(p^m-1)-p^m\wt(\vc)
\end{equation}
We distinguish between two cases

\textit{Case 1}. $\frac{af(x)}{x}=(h(x))^{p}-h(x)+b$ for some $a\in \BB\setminus\{0\}$, $h(x)\in \FF[x]$ and $b\in \FF$. In this case, the number of such $a$ is at most $p-1$ and let $\mathfrak{B}$ be the collection of such $a$. In fact, if for $a_1$ from $\mathfrak{B}$ it holds that $\frac{a_1f(x)}{x}=(h(x))^{p}-h(x)+b$, then for $a_2$ from the same set it holds that $\frac{a_2f(x)}{x}=a_2a_1^{-1}((h(x))^{p}-h(x)+b)$. Hence, $\chi_{a_2}\left(\frac{f(x)} {x}\right)$ is a constant number for each $x\in\cA$ when $a_2a_1^{-1}$ belongs to the finite field with $p$ elements. Utilizing the estimate~\eqref{weilestimate2}, we have
\begin{align*}
&\left|\sum_{a\in \BB\setminus\{0\}} \sum_{\alpha\in\cA}\chi_a\left(\frac{f(\alpha)} {\alpha}\right)\right|
\le\sum_{a\in \BB\setminus\{0\}}\left|\sum_{\alpha\in\cA}\chi_a\left(\frac{f(\alpha)} {\alpha}\right)\right|\\
&\qquad\le(p-1)\#\cA+\sum_{a\in \BB\setminus(\{0\}\cup \mathfrak{B})}\left|\sum_{\alpha\in\cA}\chi_a\left(\frac{f(\alpha)}{\alpha}\right)\right|\\
&\qquad\le(p-1)(p^{mt}-1)+(p^m-p)(k-1)\sqrt{p^{mt}}.
\end{align*}
By \eqref{prop:eq1}, we obtain that $$\wt(\mathbf{c})\geq \frac{(p^m-p)\left((p^{mt}-1)-(k-1)\sqrt{p^{mt}}\right)}{p^m}.$$

\textit{Case 2}. $f(x)\neq x(h(x))^{p}-xh(x)+xb$ for any $a\in \BB\setminus\{0\}$, $h(x)\in \FF[x]$ and $b\in \FF$. Using a method analogous to \textit{Case 1}, we deduce that
\begin{align*}
\wt(\vc)\geq \frac{(p^m-1)\left((p^{mt}-1)-(k-1)\sqrt{p^{mt}}\right)}{p^m}.    
\end{align*}
Taking the minimum over two cases, we get the proposition statement. 
\end{proof}

Combining Proposition~\ref{prop:char-sum} and Proposition~\ref{prop:char-sum-new} together, we get Theorem~\ref{thm:distance_lower_bound}(iii).

\section{Robust Repair Schemes }
\label{sec:decoders}

\subsection{Robust Repair Scheme 1}
\label{subsec:decoder1}

In this section, we propose Robust Repair Scheme 1 that achieves the error-correction capability guaranteed by the BCH-bound (Algorithm~D). 
Suppose we receive a noisy word $\boldsymbol{y} = \boldsymbol{c}_{\text{true}} + \boldsymbol{e}$, where $\boldsymbol{c}_{\text{true}} \in \cRT(\FF^*, k)$ is the true repair-trace codeword and the error vector $\boldsymbol{e}$ has Hamming weight $\wt(\boldsymbol{e}) \le \lfloor(d_{\text{BCH}}-1)/2\rfloor$. From Algorithm~D, we identify the parameters $b_{\text{opt}}$ and $s_{\text{opt}}$ that locate a sequence of $\delta_{\text{opt}}$ many $b_{\text{opt}}^{-1}$--spaced consecutive zero coefficients in the corresponding repair-trace polynomial $F(x)$. Furthermore, with $a \equiv q^t - \delta_{\text{opt}} - s_{\text{opt}} - 1 \pmod{q^t - 1}$, the polynomial $F^{(a,b_{\text{opt}})}(x) \equiv x^a F(x^{b_{\text{opt}}}) \pmod{x^n - 1}$ contains $\delta_{\text{opt}}$ consecutive zero coefficients at the highest-degree terms, specifically at positions $q^t - \delta_{\text{opt}} - 1$ through $q^t - 2$. Our task is to apply a bijective transformation on the received word that mirrors the polynomial transformation, such that the uncorrupted symbols in the transformed word are governed by $F^{(a,b_{\text{opt}})}(x)$.

Specifically, for each evaluation point $\alpha \in \FF^*$, we compute
$$T^{(a,b_{\text{opt}})}(y_{\alpha}) = \alpha^a y_{\alpha^{b_{\text{opt}}}},$$
with the shift parameter $a \equiv q^t - \delta_{\text{opt}} - s_{\text{opt}} - 1 \pmod{q^t - 1}$. Because $\gcd(b_{\text{opt}}, q^t - 1) = 1$, the mapping $x \mapsto x^{b_{\text{opt}}}$ acts as a coordinate permutation over $\FF^*$, and multiplication by $\alpha^a$ is a non-zero scalar multiplication. Consequently, $T^{(a,b_{\text{opt}})}$ is bijective and preserves the Hamming weight of the error vector. We can therefore express the transformed word as
$$T^{(a,b_{\text{opt}})}(\boldsymbol{y}) = T^{(a,b_{\text{opt}})}(\boldsymbol{c}_{\text{true}}) + \boldsymbol{e}',$$
where the transformed error vector $\boldsymbol{e}'$ satisfies $\wt(\boldsymbol{e}') = \wt(\boldsymbol{e}) \le \lfloor(d_{\text{BCH}}-1)/2\rfloor$.
Since the transformed true codeword $T^{(a,b_{\text{opt}})}(\boldsymbol{c}_{\text{true}})$ corresponds to a polynomial $F^{(a,b_{\text{opt}})}(x)$ of degree at most $q^t - \delta_{\text{opt}} - 2$, it is a codeword in $\rs(\FF^*, q^t - \delta_{\text{opt}} - 1)$. This allows us to feed the transformed word directly into the Berlekamp-Welch algorithm to uniquely correct up to $\lfloor\delta_{\text{opt}}/2\rfloor$ errors and obtain $F^{(a,b_{\text{opt}})}(x)$.

Once the transformed polynomial $F^{(a,b_{\text{opt}})}(x)$ is successfully decoded, we reverse the transformation to recover the original repair-trace polynomial $F(x)$, that is, $F(x) = x^{-ab_{\text{opt}}^{-1}} F^{(a,b_{\text{opt}})}(x^{b_{\text{opt}}^{-1}}) \pmod{x^n - 1}$. Note that our ultimate goal is to obtain $f(0) = f_0$. To perform the trace-repair scheme, we require $k \le q^t - q^{t-1}$. Consequently, even in the worst-case scenario where $k = q^t - q^{t-1}$, $f_0$ remains the only coefficient in $f(x)$ whose corresponding terms in $F(x)$ fall into the cyclotomic coset $C_{r_{|\Xi|}}$. We formalize this in Lemma~\ref{lem:f0_extraction}. This observation allows us to extract $f_0$ directly from the recovered $F(x)$ without performing the standard trace-dual basis reconstruction. We formalize this procedure as Robust Repair Scheme 1.

\begin{lemma}
\label{lem:f0_extraction}
For any $k \le q^t - q^{t-1}$, the erased symbol $f(0) = f_0$ can be extracted directly from the coefficient of $x^{q^t-2}$ in the recovered trace polynomial $F(x)$.
\end{lemma}

\begin{proof}
By Lemma \ref{lem:largest_rep}, the largest coset representative is $r_{|\Xi|} = q^t - q^{t-1} - 1$. The trace polynomial is expanded as
\begin{equation}
    F(x) = \text{Tr}(f_0/x) + \sum_{i=0}^{k-2} \text{Tr}(f_{i+1}x^i).
\end{equation}
Since $k \le q^t - q^{t-1}$, then $k - 2 \le q^t - q^{t-1} - 2 < r_{|\Xi|}$. Therefore, by Theorem~\ref{thm:sup_is_union_of_cyc_coset}, none of the terms $f_{i+1}x^i$ can belong to the cyclotomic coset $C_{r_{|\Xi|}}$. The only term contributing to $C_{r_{|\Xi|}}$ is $\text{Tr}(f_0/x) = \sum_{i=0}^{t-1} f_0^{q^i} x^{(q^t-2)q^i}$. Evaluating this sum at $i=0$ yields exactly $f_0 x^{q^t-2}$. Thus, the coefficient of $x^{q^t-2}$ in $F(x)$ has no overlapping terms and is equal to $f_0$.
\end{proof}

\begin{figure}[htbp]
    \centering
    \fbox{
        \begin{minipage}{0.45\textwidth}
        \vspace{0.5em}
        \noindent \textbf{Robust Repair Scheme 1}

        \vspace{0.5em}
        \noindent\textbf{Input:} Received word $\mathbf{y} = (y_\alpha)_{\alpha \in \mathbb{F}^*}$, code length $n = q^t - 1$, support $\cS_k$.
        
        \noindent\textbf{Output: } The erased symbol $f_0$, or ``Decoding Failure''.

        \begin{enumerate}[label = \textbf{1.\arabic*.}, leftmargin=*, labelsep = 0.5em]
            \item \textbf{(Preprocessing Phase)} Run Algorithm D to obtain $b_{\text{opt}}$, $\delta_{\text{opt}}$ and $s_{\text{opt}}$. Compute the shift parameter $a \gets n - \delta_{\text{opt}} - s_{\text{opt}} \pmod n$.
            \item \textbf{(Transformation)} For each $\alpha \in \mathbb{F}^*$, compute the transformed evaluation $y_{\alpha}^{(a,b_{\text{opt}})} \gets \alpha^a y_{\alpha^{b_{\text{opt}}}}$.
            
            \item \textbf{(Decoding Phase)} Run the Berlekamp-Welch algorithm 
            on the transformed word $(y_{\alpha}^{(a,b_{\text{opt}})})_{\alpha \in \mathbb{F}^*}$ under the assumption that the number of errors $e \le \lfloor \delta_{\text{opt}}/2 \rfloor$. If the algorithm fails to output a valid polynomial $P(x)$, terminate and return ``Decoding Failure''.
            
            \item \textbf{(Reconstruction)} Reconstruct the original repair trace polynomial $F(x) \gets x^{-a b_{\text{opt}}^{-1}} P(x^{b_{\text{opt}}^{-1}}) \pmod{x^n - 1}$. Extract and return $f_0$ from the coefficient of $x^{q^t-2}$ in $F(x)$.
        \end{enumerate}
        \vspace{0.5em}
        \end{minipage}
    }
    \label{fig:decoder_1}
\end{figure}

\begin{theorem}[Correctness of Robust Repair Scheme 1]
    Let $\boldsymbol{c}_{\text{true}} \in \cRT(\FF^*, k)$ be the true repair-trace codeword and $\boldsymbol{y} = \boldsymbol{c}_{\text{true}} + \boldsymbol{e}$ be the received word, with the error vector $\boldsymbol{e}$ satisfying $\wt(\boldsymbol{e}) \le \lfloor(d_{\text{BCH}} - 1)/2\rfloor$. Then, Robust Repair Scheme 1 uniquely and correctly outputs the erased symbol $f_0$.
\end{theorem} 
\begin{proof}
   By construction, the bijective transformation preserves the error weight, mapping the received word into a corrupted Reed-Solomon codeword with exactly $\wt(\boldsymbol{e})$ errors. Since this weight is bounded by the unique decoding radius $\lfloor\delta_{\text{opt}}/2\rfloor$, the Berlekamp-Welch algorithm is guaranteed to recover the transformed version of $\boldsymbol{c}_{\text{true}}$. Then, the inverse transformation reconstructs $\boldsymbol{c}_{\text{true}}$, and by Lemma 17, $f_0$ is correctly repaired.
\end{proof}

\subsection{Robust Repair Scheme 2}

Although Robust Repair Scheme 1 efficiently corrects up to $\lfloor (d_{\text{BCH}}-1)/2 \rfloor$ erroneous traces, there remains a gap between this number and the theoretical number of errors guaranteed by the Character Sum bound, $\lfloor (d_{\text{CS}}-1)/2 \rfloor$, for certain dimensions $k$. For instance, as illustrated in Fig.~1, when $p=2, m=1, t=8,$ and $k \in [3, 13]$, the error-correction capability guaranteed by the Character Sum bound strictly dominates the BCH bound. Our objective with Robust Repair Scheme 2 is to close this gap by manipulating the support of the repair-trace polynomial. As established in Theorem~\ref{thm:sup_is_union_of_cyc_coset}, the support of the repair-trace polynomial $F(x)$ is a union of cyclotomic cosets and some of these cyclotomic cosets are involved in the support of $F(x)$ due to exactly one coefficient of $f(x)$. 

\begin{definition}
    A cyclotomic coset $C\subset \cS_k$ is said to be \textit{isolated} if the terms in $F(x) = \tr(f(x)/x)$ with exponents in $C$ depend on exactly one coefficient of $f(x)$.
\end{definition}

For example, in Example~\ref{ex:support_of_F}, with $\mathbb{F} = GF(2^5)$ and $\mathbb{B} = GF(2)$, the isolated cyclotomic cosets are $C_0, C_5, C_7,$ and $C_{15}$. In particular, the terms in $F(x)$ corresponding to $C_{15}$ are constructed entirely by a single coefficient $f_0$. Because the terms within an isolated coset $C$ are dictated by a single coefficient $c\in\FF$, we can exhaustively guess the value of this coefficient. Since there are at most $|\FF|$ possible values, this guessing phase remains computationally tractable. For the correct guess $c$, the modified polynomial $F_c(x) = F(x) - \tr(c x^r)$, with $r$ being the representative of $C$, will have a strictly smaller support $\mathcal{S}_k \setminus C$. This transformation enlarges the pool of zero coefficients, which yields a larger BCH bound for the transformed word $(y_{\alpha} - \tr(c\alpha^{r_C}))_{\alpha\in\FF^*}$ corresponding to $F_c(x)$. Incorrect guesses for $c$ are reliably eliminated by enforcing a consistency check on the support of the decoded polynomial. However, because this newly enlarged BCH bound might exceed some theoretical distance lower bound of the original code, say $d$, we must strictly cap the decoding radius at $\lfloor (d-1)/2 \rfloor$ to guarantee the uniqueness of the solution. For Robust Repair Scheme~2, we take $d = d_{\text{CS}}$, the character sum bound. Furthermore, since we are unable to guarantee uniqueness beyond the number of correctable errors guaranteed by the character sum bounds, we simply apply Robust Repair Scheme~1 for $k$ where the BCH bound is greater than the character sum bound.

We formalize this idea in Robust Repair Scheme~2 and prove its correctness and uniqueness in Theorem~\ref{thm:uniqueness_of_decoder_2}.

\begin{figure}[htbp]
    \centering
    \fbox{
        \begin{minipage}{0.45\textwidth}
        \vspace{0.5em}
        \noindent \textbf{Robust Repair Scheme 2}

        \vspace{0.5em}
        \noindent\textbf{Input:} Received word $\mathbf{y} = (y_\alpha)_{\alpha \in \mathbb{F}^*}$, code length $n = q^t - 1$, support $\cS_k$, character sum bound $d_{\text{CS}}$.
        
        \noindent\textbf{Output: } The erased symbol $f_0$, or a decoding failure flag.

        \begin{enumerate}[label = \textbf{2.\arabic*.}, leftmargin=*, labelsep = 0.5em]
            \item \textbf{(Comparison Phase)} Run Algorithm D with support $\cS_k$ to obtain $d_{\text{BCH}}$. If $\lfloor (d_{\text{BCH}}-1)/2\rfloor \ge \lfloor (d_{\text{CS}}-1)/2\rfloor$, then run Decoder 1. Else, proceed to \textbf{2.2}.
            \item \textbf{(Preprocessing Phase)} 
            \begin{enumerate}[label = \textbf{2.1.\arabic*.}, leftmargin = 2.8em, labelsep = 0.5em]
                \item Set $I\gets \{C: C\subset \cS_k \text{ is isolated}\}$, $\delta_{\text{best}} \gets 0$.
                \item For each $C\in I$,
                \begin{enumerate}[label = \textbf{2.1.2.\arabic*.}, leftmargin = 3.6em, labelsep = 0.5em]
                    \item Set $\cS_k^* \gets \cS_k\setminus C$.
                    \item Run Algorithm D with support $\cS_k^*$ to obtain $\delta_{\text{opt}}, b_{\text{opt}}$ and $s_{\text{opt}}$.
                    \item If $\delta_{\text{opt}}>\delta_{\text{best}}$, set $\delta_{\text{best}}\gets \delta_{\text{opt}}$, $ b_{\text{best}}\gets b_{\text{opt}}$, $s_{\text{best}}\gets s_{\text{opt}}$, $ \tilde{C} \gets C$.
                \end{enumerate}
                \item Compute $a \gets n-\delta_{\text{best}} - s_{\text{best}}\pmod{n}$ .
            \end{enumerate}
        \item \textbf{(Testing Phase)} For each $c\in \FF$,
        \begin{enumerate}[label = \textbf{2.2.\arabic*.}, leftmargin = 2.8em, labelsep = 0.5em]
            \item \textbf{(Transformation Phase)} For each $\alpha\in \FF^*$,
            \begin{enumerate}[label = \textbf{2.2.1.\arabic*.}, leftmargin = 3.6em, labelsep = 0.5em]
                \item Set $y_{\alpha,c}\gets y_{\alpha} - \tr(c\alpha^{r_{\tilde{C}}})$.
                \item Set $y_{\alpha,c}^{(a,b_{\text{best}})}\gets \alpha^a y_{\alpha^{b_{\text{best}}},c}$
            \end{enumerate}
        \item \textbf{(Decoding Phase)} Run the Berlekamp-Welch algorithm on $\left(y_{\alpha,c}^{(a,b_{\text{best}})}\right)_{\alpha\in\FF^*}$ under the assumption that the number of errors $e \le \min\{\lfloor\delta_{\text{best}}/2\rfloor, \lfloor(d_{\text{CS}}-1)/2\rfloor\} $. Let $F_c(x)$ be the output. If it fails to output a valid polynomial $F_c(x)$, then guess another $c$.
        \item \textbf{(Reconstruction)} Reconstruct the original repair trace polynomial $$F(x) \gets x^{-ab^{-1}_{\text{best}}}F_c(x^{b^{-1}_{\text{best}}}) + \tr(c x^{r_{\tilde{C}}}),$$
        in $\textrm{mod }{x^n-1}$. If $\supp(F)\subseteq \cS_k$ and $F_{iq(\text{mod }q^t-1)} = F_i^q$, then extract and return $f_0$ from the coefficient of $x^{q^t-2}$ in $F(x)$, and terminate \textbf{2.3}.
        \end{enumerate}
        If no candidate $c \in \FF$ produces a valid $F(x)$, return ``Decoding Failure''.
        \end{enumerate}
        \vspace{0.5em}
        \end{minipage}
    }
    \label{fig:decoder_2}
\end{figure}
\begin{theorem}\label{thm:uniqueness_of_decoder_2}
    Let $\boldsymbol{c}_{\text{true}} \in \cRT(\FF^*, k)$ be the true repair-trace codeword and $\boldsymbol{y} = \boldsymbol{c}_{\text{true}} + \boldsymbol{e}$ be the received word with an error vector $\boldsymbol{e}$. Let $d_{BCH, C}$ be the BCH bound derived from the reduced support $\mathcal{S}_{k} \setminus C$ for some isolated coset $C$. If the error vector $\boldsymbol{e}$ satisfies $\wt(\boldsymbol{e}) \le \min\left\{\lfloor(d_{BCH, C} - 1)/2\rfloor, \lfloor(d_{CS} - 1)/2\rfloor\right\}$, then Robust Repair Scheme 2 uniquely and correctly outputs the erased symbol $f_0$.
\end{theorem}
\begin{proof}
    In the comparison phase, if the BCH bound guarantees more error correction than the character sum bound, we run Robust Repair Scheme~1. In other words, we only need to show correctness and uniqueness for the case $d_{\text{BCH}} < d_{\text{CS}}$. The preprocessing phase outputs the best isolated cyclotomic coset $C$ that maximizes the BCH bound and also outputs the required parameters $\delta_{\text{best}}, b_{\text{best}}, $and $s_{\text{best}}$. We use the outputs from the preprocessing phase and prove, Robust Repair Scheme~2 outputs $f_0$ correctly with the correct guess $c'$ and rejects all outputs from wrong guesses $c\ne c'$.
    
    \textbf{(Correctness)} For the correct guess, say $c'$, let $\boldsymbol{c}_{\text{true},c'} = \boldsymbol{c}_{\text{true}} -  (\tr(c' \alpha^{r_C}))_{\alpha\in\FF^*}$. Then,  we have
    \begin{align*}
        \boldsymbol{y}_{c'} &= \boldsymbol{y} - (\tr(c' \alpha^{r_C}))_{\alpha\in\FF^*} = \boldsymbol{c}_{\text{true}} -  (\tr(c' \alpha^{r_C}))_{\alpha\in\FF^*} + \boldsymbol{e}\\ 
        &= \boldsymbol{c}_{\text{true},c'}  + \boldsymbol{e},
    \end{align*}
    which shows that such transformation preserve the error weight. Furthermore, the transformed codeword $\boldsymbol{c}_{\text{true},c'}$ corresponds to a polynomial $F_{c'}(x)$ with support set in $\cS_k\setminus C$. Following the idea of Robust Repair Scheme~1, we further transform $\boldsymbol{y}_{c'}$ into $T^{(a,b_{\text{best}})}(\boldsymbol{y}_{c'})$, where
    $$
    T^{(a,b_{\text{best}})}(\boldsymbol{y}_{c'}) = \left( T^{(a,b_{\text{best}})}({y}_{\alpha,c'})\right)_{\alpha\in\FF^*} = \left(\alpha^a y_{\alpha^{b_{\text{best}}},c'}\right)_{\alpha\in\FF^*}
    $$
    with $a \equiv q^t -\delta_{\text{best}} - s_{\text{best}} - 1 \pmod{q^t-1}$. Such transformation is bijective, hence it preserves the error weight. That is,
    $$
    T^{(a,b_{\text{best}})}(\boldsymbol{y}_{c'}) = T^{(a,b_{\text{best}})}(\boldsymbol{c}_{\text{true},c'}) + T^{(a,b_{\text{best}})}(\boldsymbol{e}),
    $$
    with $\wt\left(T^{(a,b_{\text{best}})}(\boldsymbol{e})\right) = \wt(\boldsymbol{e})$.
    Furthermore, since $T^{(a,b_{\text{best}})}(\boldsymbol{c}_{\text{true},c'})$ is a codeword corresponding to a polynomial $F^{(a,b_{\text{best}})}_{c'}(x)$ of degree at most $q^t - \delta_{\text{best}} - 2$, it is a codeword in $\rs(\FF^*, q^t - \delta_{\text{best}} - 1)$. This allows us to feed $T^{(a,b_{\text{best}})}(\boldsymbol{y}_{c'})$ directly into the Berlekamp-Welch algorithm to uniquely correct up to $\lfloor \delta_{\text{best}}/2 \rfloor$ errors and obtain $F^{(a,b_{\text{best}})}_{c'}(x)$. Then, the Reconstruction Phase computes $F(x)$ by performing inverse transformation of $T^{(a,b_{\text{best}})}$ and adding the trace terms back which results in the corresponding polynomial of $\boldsymbol{c}_{\text{true}}$. Since $\boldsymbol{c}_{\text{true}} \in\cRT(\FF^*,k)$, it clearly passes the check. Finally, by Lemma~\ref{lem:f0_extraction}, $f_0$ is correctly repaired.
    
    \textbf{(Uniqueness)} Now, we show that no incorrect guess $c\ne c'$ can result in a valid polynomial corresponding to $\cRT(\FF^*,k)$ through the check. Suppose that the decoder with the correct guess $c'$ outputs $F(x)$ corresponding to the valid codeword $\boldsymbol{c}_{\text{true}}\in\cRT(\FF^*,k)$. Furthermore, suppose that, with an incorrect guess $c\ne c'$, in step 2.2.3, we also obtain another polynomial $F_{\text{false}}(x)$ corresponding to $\boldsymbol{c}_{\text{false}}\in\cRT(\FF^*,k)$. Set $e = \min\left\{\lfloor(d_{BCH, C} - 1)/2\rfloor, \lfloor(d_{CS} - 1)/2\rfloor\right\}$. Then, we have $\wt(\boldsymbol{c}_{\text{true}} - \boldsymbol{y})\le e$ and $\wt(\boldsymbol{c}_{\text{false}} - \boldsymbol{y})\le e$. Since $\cRT(\FF^*,k)$ is a linear code, $\Delta \boldsymbol{c} = \boldsymbol{c}_{\text{true}} - \boldsymbol{c}_{\text{false}}$ is a valid codeword of $\cRT(\FF^*,k)$ and $\wt(\Delta \boldsymbol{c}) \ge d_{\text{CS}}$. But,
    \begin{align*}\wt(\Delta\boldsymbol{c}) &= \wt((\boldsymbol{c}_{\text{true}} - \boldsymbol{y}) - (\boldsymbol{c}_{\text{false}} - \boldsymbol{y}))\\
    &\le \wt(\boldsymbol{c}_{\text{true}} - \boldsymbol{y}) + \wt(\boldsymbol{c}_{\text{false}} - \boldsymbol{y})\le 2e < d_{\text{CS}},
    \end{align*}
    which is a contradiction.
\end{proof}

\subsection{Robust Repair Scheme 2 with List Decoding}

Again, as shown in Fig.~\ref{fig:comparison_of_decoders_and_distance_bound}, there remains a gap between the number of erroneous traces correctable by Robust Repair Scheme~2 and the number guaranteed by the character sum bound for some values $k$. To further close this gap, we replace the Berlekamp-Welch decoding algorithm with the Guruswami-Sudan List decoding algorithm \cite{GuruswamiSudan1998} in Robust Repair Scheme~2 step 2.2.2. to correct for a larger decoding radius. As before, to guarantee the uniqueness of the output, we must cap the decoding radius at $\lfloor (d_{\text{CS}}-1)/2\rfloor$.

We quote the performance of Guruswami-Sudan algorithm due to Koetter and Vardy \cite{koettervardy2003} in Theorem~\ref{thm:GS_list_decoding}.

\begin{theorem}[\cite{GuruswamiSudan1998, koettervardy2003}]\label{thm:GS_list_decoding}
	Fix $\Delta\le n$ and $\mu$.
	Set $\delta$ to be the smallest integer such that 
	$N_{1,\Delta}\triangleq \ceil{\frac{\delta+1}{\Delta}}\left(\delta-\frac{\Delta}{2}\floor{\frac{\delta}{\Delta}}+1\right)>\frac{n\mu(\mu+1)}{2}$.
	Next, set $e = n-\floor{\delta/\mu}$.
	Given $\vy\in \FF^n$, we can find {\em all} polynomials $F(X)$ of degree at most $\Delta$ such that 
	$F(\alpha_i) \ne y_i$ in at most $e$ positions.
	Let $\cF$ be the set of these polynomials. Furthermore, we can find all $F(X)$'s in polynomial  time (in $n$, $\mu$ and $|\cF|$).
\end{theorem}

In \cite{guruswami2012}, it is shown that when $\mu = 2\Delta n$, the Guruswami-Sudan algorithm can perform efficient list decoding with up to $e\le n-\left\lceil \sqrt{\Delta n+1/2}\right\rceil$
errors, where $\Delta$ is the degree of the corresponding polynomial. In other words, when $d_{\text{BCH}} < d_{\text{CS}}$, the degree of polynomial of interest in Robust Repair Scheme~2 is of degree $n-\delta_{\text{best}} -1$ where $\delta_{\text{best}}$ is an output from the preprocessing phase (step 2.2). Hence, Robust Repair Scheme~2 with list decoding can correct at most $e$ erroneous traces, where
\begin{equation*}
e = \begin{cases}
    \begin{aligned}
        &\min\Bigl\{n-\bigl\lceil \sqrt{(n-\delta_{\text{best}} -1)n+1/2}\bigr\rceil, \\
        &\qquad\quad \lfloor(d_{\text{CS}}-1)/2\rfloor\Bigr\} 
    \end{aligned} & \text{if } d_{\text{BCH}} < d_{\text{CS}},\\
    \lfloor (d_{\text{BCH}}-1)/2 \rfloor & \text{if } d_{\text{BCH}} \ge d_{\text{CS}}.
\end{cases}
\end{equation*}
As shown in Fig.~\ref{fig:comparison_of_decoders_and_distance_bound} that Robust Repair Scheme~2 with list decoding improves the number of tolerable erroneous traces from Robust Repair Scheme~2 with unique decoder. While the use of list decoding further increases the number of correctable errors, the gain is modest in practice.

\section{Conclusion}

In this paper, we studied the problem of bandwidth-efficient repair of a single erasure in full-length Reed–Solomon codes over small and moderate field sizes in the presence of erroneous traces. We established upper bounds on the dimension of Reed–Solomon codes that enable correction of a given number of errors, and derived an explicit and optimal bound for correcting at least one error in the binary case. Furthermore, we analyzed the minimum distance of the code formed by traces and provided theoretical lower bounds for it. To support practical applications, we proposed two robust repair schemes. Extending these results to non-full-length Reed–Solomon codes remains an interesting direction for future research.

\appendix

\begin{proof}[Proof of Lemma~\ref{lem:binary_rep}]
    Lemma~\ref{lem:binary_rep}.1. is a direct consequence of Lemma~\ref{lem:largest_rep} with $q = 2$. Note that, we can represent $r_{|\Xi|}$ in binary representation $0\boldsymbol{1}_{t-1}$ and the multiplication by $2$ results in a cyclic shift in any elements in the cyclotomic coset. For instance, $2r_{|\Xi|}$ is $\boldsymbol{1}_{t-1}0, 2^2 r_{|\Xi|}$ is $\boldsymbol{1}_{t-2}01$, and so on. This also shows that $|C_{r_{|\Xi|}}| = t$ and all strings with only one $0$ are in $C_{r_{|\Xi|}}$.
    
    \textbf{Claim:} The second largest representative $r_{|\Xi|-1}$ has binary representation containing exactly two $0$'s.
    
    \textit{Proof of Claim.} To prove this, we show that the largest binary representation containing two $0$'s, $0\boldsymbol{1}_{a_2}0\boldsymbol{1}_{b_2}$, is larger than the largest binary representation containing three $0$'s, $0\boldsymbol{1}_{a_3}0\boldsymbol{1}_{b_3}0\boldsymbol{1}_{c_3}$. For $0\boldsymbol{1}_{a_2}0\boldsymbol{1}_{b_2}$ to be the largest, we maximize $a_2$ under the constraints $a_2 + b_2 = t-2$ (since there are in total $t$ slots) and $0\le a_2 \le b_2$ (since it is a representative). Hence, we require $a\le (t-2)/2$ and the largest integer $a_2$ is $\lfloor t/2\rfloor -1$. Now, for $0\boldsymbol{1}_{a_3}0\boldsymbol{1}_{b_3}0\boldsymbol{1}_{c_3}$ to be the largest, in particular, we maximize $a_3$ under the constraints  $a_3 + b_3 + c_3 = t-3$ and $0\le a_3 \le b_3,c_3$. Hence, we require $a\le (t-3)/3$ and the largest integer $a_3$ is $\lfloor t/3\rfloor - 1$. Clearly, for $t\ge 3$, $a_3 =\lfloor t/3\rfloor - 1  \le \lfloor t/2\rfloor - 1 = a_2$. Hence,
    \begin{align*}
0\boldsymbol{1}_{a_3}0\boldsymbol{1}_{b_3}0\boldsymbol{1}_{c_3} \le 0\boldsymbol{1}_{a_3}0\boldsymbol{1}_{t-a_3 - 2} \le 0\boldsymbol{1}_{a_2}0\boldsymbol{1}_{b_2}.
    \end{align*}
    
    Hence, the binary representative of $r_{|\Xi|-1}$ must contain two zeros and based on the proof of the claim above, it is of the form $0\boldsymbol{1}_{\lfloor t/2\rfloor-1} 0 \boldsymbol{1}_{\lceil t/2\rceil-1}$. For integer $t$, $\lceil t/2\rceil-1 = \lfloor (t-1)/2\rfloor$ and the result $r_{|\Xi|-1} = 2^{t-1} - 2^{\lfloor(t-1)/2\rfloor} - 1$, for $t\ge 3$ follows.

    Clearly, the next largest representative containing two $0$'s is $0\boldsymbol{1}_{\lfloor t/2\rfloor - 2} 0 \boldsymbol{1}_{\lceil t/2\rceil}$. We are left to check that the largest representative containing three $0$'s is never larger than it. Indeed, for $t\ge 4$, $a_3 = \lfloor t/3\rfloor -1 \le \lfloor t/2 \rfloor - 2 = a_2 - 1$. Hence,
\begin{align*}
0\boldsymbol{1}_{a_3}0\boldsymbol{1}_{b_3}0\boldsymbol{1}_{c_3} \le 0\boldsymbol{1}_{a_3}0\boldsymbol{1}_{t-a_3 - 2} \le 0\boldsymbol{1}_{a_2-1}0\boldsymbol{1}_{b_2+1}.
    \end{align*}
    Therefore, $r_{|\Xi|-2} = 2^{t-1} - 2^{\lfloor(t-1)/2\rfloor+ 1} - 1$, for $t\ge 4$.
\end{proof}
\balance
\printbibliography

\end{document}